\documentclass{article}
\usepackage{fullpage}
\usepackage{microtype}
\usepackage{graphicx}
\usepackage{subfigure}
\usepackage{booktabs} 

\usepackage{etoolbox}

\usepackage[utf8]{inputenc} 
\usepackage[T1]{fontenc}    
\usepackage{hyperref}       
\usepackage{url}            
\usepackage{booktabs}       
\usepackage{amsfonts}       
\usepackage{nicefrac}   
\usepackage[square,sort,comma,numbers]{natbib}
\usepackage{microtype}      
\usepackage{todonotes}
\usepackage{microtype,enumitem}
\usepackage{mathtools} 
\usepackage{amsmath}
\usepackage{bm}
\usepackage{algorithm2e}
\usepackage{algorithm}
\usepackage{amsthm}
\pdfoutput=1

\newtheorem{prop}{Proposition}
\newtheorem{thm}{Theorem}

\newtheorem{assump}{Assumption}
\newtheorem{corollary}{Corollary}
\newtheorem{lemma}{Lemma}
\newtheorem{definition}{Definition}

\newtheorem{rmq}{Remark}

\newtheorem*{assump*}{Assumption}

\newtheorem*{example*}{Example}
\newtheorem*{prop*}{Proposition}
\newtheorem*{thm*}{Theorem}
\newtheorem*{prv*}{Proof}

\newcommand{\risk}{\mathcal{R}}
\newcommand{\riskadv}{\mathcal{R}_{adv}}
\newcommand{\valuerand}{\mathcal{V}_{rand}}
\newcommand{\valuedet}{\mathcal{V}_{det}}
\newcommand{\dualvalue}{\mathcal{D}}

\DeclareMathOperator*{\supp}{supp}   
   
\DeclareMathOperator*{\essinf}{essinf}   

\newcommand{\QQ}{\mathbb{Q}}
\newcommand{\PP}{\mathbb{P}}

\DeclareMathOperator*{\argmaxB}{argmax} 

\usepackage{wrapfig, framed, caption}
\usepackage{multirow}

\title{Mixed Nash Equilibria in the Adversarial Examples Game}

\author{
	Laurent Meunier$^{*1,2}$\qquad Meyer Scetbon$^{*3}$\qquad Rafael Pinot$^{4}$\\
	Jamal Atif$^{1}$ \qquad
	Yann Chevaleyre$^{1}$	\vspace{0.3cm}
\\
	$^{1}$ LAMSADE, Université Paris-Dauphine\\
	$^{2}$ Facebook AI Research, Paris\\
	$^{3}$ CREST, ENSAE\\
	$^{4}$ Ecole Polytechnique Fédérale de Lausanne\\

	}
\date{}

\begin{document}

\maketitle

\begin{abstract} 

This paper tackles the problem of adversarial examples from a game theoretic point of view. We study the open question of the existence of mixed Nash equilibria in the zero-sum game formed by the attacker and the classifier. While previous works usually allow only one player to use randomized strategies, we show the necessity of considering randomization for both the classifier and the attacker. We demonstrate that this game has no duality gap, meaning that it always admits approximate Nash equilibria. We also provide the first optimization algorithms to learn a mixture of classifiers that approximately realizes the value of this game, \emph{i.e.} procedures to build an optimally robust randomized classifier.

\end{abstract}

\section{Introduction}

Adversarial examples~\citep{biggio2013evasion,Szegedy2013IntriguingPO} are one of the most dizzling problems in machine learning: state of the art classifiers are sensitive to imperceptible perturbations of their inputs that make them fail. Last years, research have concentrated on proposing new defense methods~\citep{madry2017towards,moosavi2019robustness,KolterRandomizedSmoothing} and building more and more sophisticated attacks~\citep{goodfellow2014explaining,kurakin2016adversarial,carlini2017adversarial,croce2020reliable}. So far, most defense strategies proved to be vulnerable to these new attacks or are computationally intractable. This asks the following question: can we build classifiers that are robust against any adversarial attack?

A recent line of research argued that randomized classifiers could help countering adversarial attacks~\citep{pruningDefenseICLR2018,Xie2017MitigatingAE,pinot2019theoretical,NIPS2019_8443}. Along this line, \cite{pinot2020randomization} demonstrated, using game theory, that randomized classifiers are indeed more robust than deterministic ones against regularized adversaries. However, the findings of these previous works are dependent on the definition of adversary they consider. In particular, they did not investigate scenarios where the adversary also uses randomized strategies, which is essential to account for if we want to give a principled answer to the above question. Previous works studying adversarial examples from the scope of game theory investigated the randomized framework (for both the classifier and the adversary) in restricted settings where the adversary is either parametric or has a finite number of strategies~\citep{7533509,DBLP:journals/corr/abs-1906-02816,bose2021adversarial}. Our framework does not assume any constraint on the definition of the adversary, making our conclusions independent on the adversary the classifiers are facing. More precisely, we answer the following questions.





\textbf{Q1:} Is it always possible to reach a Mixed Nash equilibrium in the adversarial example game when both the adversary and the classifier can use randomized strategies?

\textbf{A1:} We answer positively to this question. First we motivate in Section~\ref{sec:adv-problem} the necessity for using randomized strategies both with the attacker and the classifier. Then, we extend the work of~\cite{pydi2019adversarial}, by rigorously reformulating the adversarial risk as a linear optimization problem over distributions. In fact, we cast the adversarial risk minimization problem as a Distributionally Robust Optimization (DRO)~\citep{blanchet2019quantifying} problem for a well suited cost function. This formulation naturally leads us, in Section~\ref{sec:nash-eq}, to analyze adversarial risk minimization as a zero-sum game. We demonstrate that, in this game, the duality gap always equals $0$, meaning that it always admits approximate mixed Nash equilibria.  

\textbf{Q2:} Can we design efficient algorithms to learn an optimally robust randomized classifier?

\textbf{A2:} To answer this question, we focus on learning a finite mixture of classifiers. Taking inspiration from robust optimization~\cite{sinha2017certifying} and subgradient methods~\cite{boyd2003subgradient}, we derive in Section~\ref{sec:algo} a first oracle algorithm to optimize over a finite mixture. Then, following the line of work of~\citep{cuturi2013sinkhorn}, we introduce an entropic reguralization which allows to effectively compute an approximation of the optimal mixture. We validate our findings with experiments on a simulated and a real image dataset, namely CIFAR-10~\cite{krizhevsky2009learning}.

\begin{figure*}[!ht]
    \centering
\includegraphics[width=0.83\textwidth]{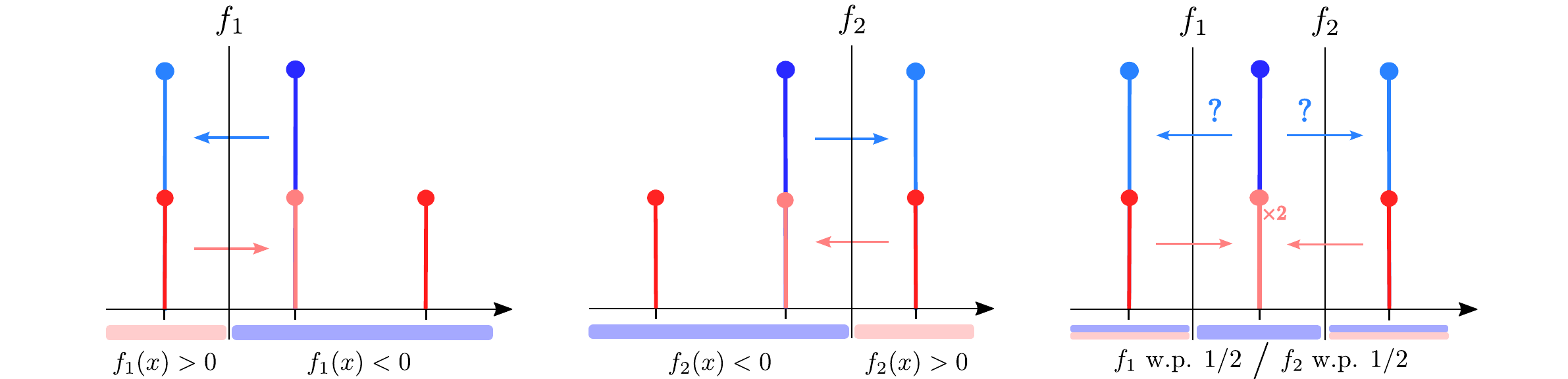}   \caption{Motivating example: blue distribution represents label $-1$ and the red one, label $+1$.  The height of columns represents their mass. The red and blue arrows represent the attack on the given classifier. On left: deterministic classifiers ($f_1$ on the left, $f_2$ in the middle) for whose, the blue point can always be attacked. On right: a randomized classifier, where the attacker has a probability $1/2$ of failing, regardless of the attack it selects.  }
    \label{fig:motivating_ex}
\vspace{-0.4cm}
\end{figure*}
\section{The Adversarial Attack Problem}
\label{sec:adv-problem}
\subsection{A Motivating Example}
\label{sec:motiv-ex}

Consider the binary classification task illustrated in Figure~\ref{fig:motivating_ex}. We assume that all input-output pairs $(X,Y)$ are sampled from a distribution $\PP$ defined as follows
$$ 
\PP\left(Y =\pm 1\right)=1/2 \ \mbox{ and }\left\{
    \begin{array}{ll}
        \PP\left(X=0 \mid Y=-1\right) = 1 \\
        \PP\left(X= \pm 1 \mid Y=1\right) = 1/2 
    \end{array}
\right.
$$ 
Given access to $\PP$, the adversary aims to maximize the expected risk, but can only move each point by at most $1$ on the real line. In this context, we study two classifiers: $f_1(x) = -x -1/2$ and $f_2(x)=x-1/2$\footnote{$(X,Y) \sim \PP$ is misclassified by $f_i$ if and only if $f_i(X)Y \leq 0$}. Both $f_1$ and $f_2$ have a standard risk of $1/4$. In the presence of an adversary, the risk (\emph{a.k.a.} the adversarial risk) increases to $1$. Here, using a randomized classifier can make the system more robust. Consider $f$ where $f=f_1$ w.p. $1/2$ and $f_2$ otherwise. The standard risk of $f$ remains $1/4$ but its adversarial risk is $3/4<1$. Indeed, when attacking $f$, any adversary will have to choose between moving points from $0$ to $1$ or to $-1$. Either way, the attack only works half of the time; hence an overall adversarial risk of $3/4$. Furthermore, if $f$ knows the strategy the adversary uses, it can always update the probability it gives to $f_1$ and $f_2$ to get a better (possibly deterministic) defense. For example, if the adversary chooses to always move $0$ to $1$, the classifier can set $f=f_1$ w.p. $1$ to retrieve an adversarial risk of $1/2$ instead of $3/4$. 

Now, what happens if the adversary can use randomized strategies, meaning that for each point it can flip a coin before deciding where to move? In this case, the adversary could decide to move points from $0$ to $1$ w.p. $1/2$ and to $-1$ otherwise. This strategy is still optimal with an adversarial risk of $3/4$ but now the classifier cannot use its knowledge of the adversary's strategy to lower the risk. We are in a state where neither the adversary nor the classifier can benefit from unilaterally changing its strategy. In the game theory terminology, this state is called a Mixed Nash equilibrium. 

\subsection{General setting}
Let us consider a classification task with input space $\mathcal{X}$ and output space $\mathcal{Y}$. Let $(\mathcal{X},d)$ be a proper (i.e. closed balls are compact) Polish (i.e. completely separable) metric space representing the inputs space\footnote{For instance, for any norm $\lVert\cdot\rVert$,  $(\mathbb{R}^d,\lVert\cdot\rVert)$ is a proper Polish metric space.}. Let $\mathcal{Y}=\{1,\dots,K\}$ be the labels set, endowed with the trivial metric  $d'(y,y') = \mathbf{1}_{y\neq y'}$. Then the space $(\mathcal{X}\times\mathcal{Y},d\oplus d')$ is a proper Polish space. For any Polish space $\mathcal{Z}$, we denote $\mathcal{M}_+^1(\mathcal{Z})$ the Polish space of Borel probability measures on $\mathcal{Z}$. Let us assume the data is drawn from $\PP\in\mathcal{M}_+^1(\mathcal{X}\times\mathcal{Y})$. Let $(\Theta,d_\Theta)$ be a Polish space (not necessarily proper) representing the set of classifier parameters (for instance neural networks). We also define a loss function: $l:\Theta\times (\mathcal{X}\times\mathcal{Y})\to [0,\infty)$ satisfying the following set of assumptions.
\vspace{-0.1cm}
\begin{assump}[Loss function]
\label{ass:loss}
1) The loss function $l$ is a non negative Borel measurable function. 2) For all $\theta\in\Theta$, $l(\theta,\cdot)$ is upper-semi continuous. 3) There exists $M>0$ such that for all $\theta\in\Theta$, $(x,y)\in\mathcal{X}\times\mathcal{Y}$, $0\leq l(\theta,(x,y))\leq M$.
\end{assump}
\vspace{-0.25cm}
It is usual to assume upper-semi continuity when studying optimization over distributions~\cite{villani2003topics,blanchet2019quantifying}. Furthermore, considering bounded (and positive) loss functions is also very common in learning theory~\cite{bartlett2002rademacher} and is not restrictive. 

In the adversarial examples framework, the loss of interest is the $0/1$ loss, for whose surrogates are misunderstood~\citep{pmlr-v97-cranko19a,pmlr-v125-bao20a}; hence it is essential that the $0/1$ loss satisfies Assumption~\ref{ass:loss}. In the binary classification setting (\emph{i.e.} $\mathcal{Y}=\{-1,+1\}$) the $0/1$ loss writes $l_{0/1}(\theta,(x,y)) = \mathbf{1}_{yf_\theta(x)\leq 0}$. Then, assuming that for all $\theta$, $f_\theta(\cdot)$ is continuous and for all $x$, $f_\cdot(x)$ is continuous, the $0/1$ loss satisfies Assumption~\ref{ass:loss}. In particular, it is the case for neural networks with continuous activation functions.

\subsection{Adversarial Risk Minimization}
The standard risk for a single classifier $\theta$ associated with the loss $l$ satisfying Assumption~\ref{ass:loss} writes: $\risk(\theta):=\mathbb{E}_{(x,y)\sim \PP}\left[l(\theta,(x,y))\right]$. Similarly, the adversarial risk of $\theta$ at level $\varepsilon$ associated with the loss $l$ is defined as\footnote{For the well-posedness, see Lemma~\ref{lem:measure-sup} in Appendix.}
\begin{align*}
    \riskadv^\varepsilon(\theta):=\mathbb{E}_{(x,y)\sim \PP}\left[\sup_{x'\in\mathcal{X},~d(x,x')\leq\varepsilon}l(\theta,(x',y))\right].
\end{align*}
 It is clear that $\riskadv^0(\theta) =\risk(\theta)$ for all $\theta$. We can generalize these notions with distributions of classifiers. In other terms the classifier is then randomized according to some distribution $\mu\in\mathcal{M}^1_+(\Theta)$. A classifier is randomized if for a given input, the output of the classifier is a probability distribution.
 The standard risk of a randomized classifier $\mu$ writes $\risk(\mu) = \mathbb{E}_{\theta\sim\mu}\left[\risk (\theta)\right]$. Similarly, the adversarial risk of the randomized classifier $\mu$ at level $\varepsilon$ is\footnote{This risk is also well posed (see Lemma~\ref{lem:measure-sup} in the Appendix).}
\begin{align*}
    \riskadv^\varepsilon(\mu):=\mathbb{E}_{(x,y)\sim \PP}\left[\sup_{x'\in\mathcal{X},~d(x,x')\leq\varepsilon}\mathbb{E}_{\theta\sim\mu}\left[l(\theta,(x',y))\right]\right].
\end{align*}
For instance, for the $0/1$ loss, the inner maximization problem, consists in maximizing the probability of misclassification for a given couple $(x,y)$. Note that $\risk(\delta_\theta)=\risk(\theta)$ and $\riskadv^\varepsilon(\delta_\theta)=\riskadv^\varepsilon(\theta)$. In the remainder of the paper, we study the adversarial risk minimization problems with randomized and deterministic classifiers and denote
\begin{align}
\label{eq:advriskmin}
    \valuerand^\varepsilon:=\inf_{\mu\in\mathcal{M}^1_+(\Theta)} \riskadv^\varepsilon(\mu),~\valuedet^\varepsilon:=\inf_{\theta\in\Theta} \riskadv^\varepsilon(\theta)
\end{align}

\begin{rmq}
We can show (see Appendix~\ref{sec:complements}) that the standard risk infima are equal :  $\valuerand^0=\valuedet^0$. Hence, no randomization is needed for minimizing the standard risk. Denoting $\mathcal{V}$ this common value, we also have the following inequalities for any $\varepsilon>0$, $\mathcal{V}\leq \valuerand^\varepsilon\leq \valuedet^\varepsilon$.
\end{rmq}

\subsection{Distributional Formulation of the Adversarial Risk} 

To account for the possible randomness of the adversary, we rewrite the adversarial attack problem as a convex optimization problem on distributions. Let us first introduce the set of adversarial distributions.
\begin{definition}[Set of adversarial distributions]
Let $\PP$ be a Borel probability distribution on $\mathcal{X}\times\mathcal{Y}$ and $\varepsilon>0$. We define the set of adversarial distributions as
\begin{align*}
\mathcal{A}_{\varepsilon}&(\PP) := \left\{\QQ\in\mathcal{M}^+_1(\mathcal{X}\times\mathcal{Y})\mid\exists \gamma\in\mathcal{M}^+_1\left((\mathcal{X}\times\mathcal{Y})^2\right),\right.\\
&\left.d(x,x')\leq\varepsilon,~y=y'~~ \gamma\text{-a.s.},~\Pi_{1\sharp}\gamma=\PP,~\Pi_{2\sharp}\gamma=\QQ\right\} 
\end{align*}
where $\Pi_i$ denotes the projection on the $i$-th component, and $g_\sharp$ the push-forward measure by a measurable function $g$.
\end{definition}
For an attacker that can move the initial distribution $\PP$ in $\mathcal{A}_\varepsilon(\PP)$, the attack would not be a transport map as considered in the standard adversarial risk. For every point $x$ in the support of $\PP$, the attacker is allowed to move $x$ randomly in the ball of radius $\varepsilon$, and not to a single other point $x'$ like the usual attacker in adversarial attacks. In this sense, we say the attacker is allowed to be randomized.

\textbf{Link with DRO. } Adversarial examples have been studied in the light of DRO by former works~\cite{sinha2017certifying,tu2018theoretical}, but an exact reformulation of the adversarial risk as a DRO problem has not been made yet. 
When $(\mathcal{Z},d)$ is a Polish space and $c:\mathcal{Z}^2\rightarrow\mathbb{R}^+\cup\{+\infty\}$  is a lower semi-continuous function, for $\PP,\QQ\in\mathcal{M}^+_1(\mathcal{Z})$ , the primal Optimal Transport problem is defined as
\begin{align*}
  W_c(\PP,\QQ):=\inf_{\gamma\in\Gamma_{\PP,\QQ}}  \int_{\mathcal{Z}^2} c(z,z')d\gamma(z,z')
\end{align*}
with $\Gamma_{\PP,\QQ}:=\left\{\gamma\in\mathcal{M}^+_1(\mathcal{Z}^2)\mid~\Pi_{1\sharp}\gamma = \PP,~\Pi_{2\sharp}\gamma = \QQ \right\}$. When $\eta>0$ and for $\PP\in\mathcal{M}^+_1(\mathcal{Z})$, the associated Wasserstein uncertainty set is defined as: 
\begin{align*}
    \mathcal{B}_{c}(\PP,\eta) := \left\{\QQ\in \mathcal{M}^+_1(\mathcal{Z})\mid W_c(\PP,\QQ)\leq \eta\right\}
\end{align*}
A DRO problem is a linear optimization problem over Wasserstein uncertainty sets $\sup_{\QQ\in\mathcal{B}_{c}(\PP,\eta)}\int g(z)d\QQ(z)$ for some upper semi-continuous function $g$~\cite{yue2020linear}. 
For an arbitrary $\varepsilon>0$, we define the cost $c_\varepsilon$ as follows
\begin{align*}
c_\varepsilon((x,y),(x',y')) := \left\{
    \begin{array}{ll}
        0 & \mbox{if } d(x,x')\leq\varepsilon\mbox{ and }y = y'\\
        +\infty & \mbox{otherwise.}
    \end{array}
\right.
\end{align*}
This cost is lower semi-continuous and penalizes to infinity perturbations that change the label or move the input by a distance greater than $\varepsilon$. As Proposition~\ref{prop:wass_ball} shows, the Wasserstein ball associated with $c_\varepsilon$ is equal to $\mathcal{A}_{\varepsilon}(\PP)$.
\begin{prop}
\label{prop:wass_ball}
Let $\PP$ be a Borel probability distribution on $\mathcal{X}\times\mathcal{Y}$ and $\varepsilon>0$ and $\eta\geq 0$, then
    $\mathcal{B}_{c_\varepsilon}(\PP,\eta) =\mathcal{A}_{\varepsilon}(\PP)$.
Moreover, $\mathcal{A}_{\varepsilon}(\PP)$ is convex and compact for the weak topology of $\mathcal{M}^+_1(\mathcal{X}\times\mathcal{Y})$.
\end{prop}
Thanks to this result, we can reformulate the adversarial risk as the value of a convex problem over $\mathcal{A}_\varepsilon(\PP)$. 
\begin{prop} 
\label{prop:dro_adv}
Let $\PP$ be a Borel probability distribution on $\mathcal{X}\times\mathcal{Y}$ and $\mu$ a Borel probability distribution on $\Theta$. Let $l:\Theta\times(\mathcal{X}\times\mathcal{Y})\to [0,\infty)$ satisfying Assumption~\ref{ass:loss}. Let $\varepsilon>0$. Then:
\begin{align}
\label{eq:dro-adv}
\riskadv^\varepsilon(\mu)= \sup_{\QQ\in \mathcal{A}_{\varepsilon}(\PP)}\mathbb{E}_{(x',y')\sim\QQ,\theta\sim\mu}\left[l(\theta,(x',y'))\right].
\end{align}
The supremum is attained. Moreover $\QQ^*\in \mathcal{A}_{\varepsilon}(\PP)$ is an optimum of Problem~\eqref{eq:dro-adv} if and only if there exists $\gamma^*\in\mathcal{M}^+_1\left((\mathcal{X}\times\mathcal{Y})^2\right)$ such that: $\Pi_{1\sharp}\gamma^*=\PP$, $\Pi_{2\sharp}\gamma^*=\QQ^*$, $d(x,x')\leq\varepsilon$, $y=y'$ and  $l(x',y')=\sup_{\substack{u\in\mathcal{X},d(x,u)\leq\varepsilon}}l(u,y)$ $\gamma^*$-almost surely.
\end{prop}

The adversarial attack problem is a DRO problem for the cost $c_\varepsilon$.
Proposition~\ref{prop:dro_adv} means that, against a fixed classifier $\mu$, the randomized attacker that can move the distribution in $\mathcal{A}_\varepsilon(\PP)$ has exactly the same power as an attacker that moves every single point $x$ in the ball of radius $\varepsilon$.  By Proposition~\ref{prop:dro_adv}, we also  deduce that the adversarial risk can be casted as a linear optimization problem over distributions.

\begin{rmq}
  In a recent work,~\citep{pydi2019adversarial} proposed a similar adversary using Markov kernels but left as an open question the link with the classical adversarial risk, due to measurability issues. Proposition~\ref{prop:dro_adv} solves these issues. The result is similar to~\citep{blanchet2019quantifying}. Although we believe its proof might be extended for infinite valued costs,~\citep{blanchet2019quantifying} did not treat that case. We provide an alternative proof in this special case. 
\end{rmq}

\section{Nash Equilibria in the Adversarial Game}
\label{sec:nash-eq}

\subsection{Adversarial Attacks as a Zero-Sum Game}

Thanks to Proposition~\ref{sec:adv-problem}, the adversarial risk minimization problem can  be seen as a two-player zero-sum game that writes as follows,
\begin{align}
    \inf_{\mu\in\mathcal{M}^1_+(\Theta)} \sup_{\QQ\in\mathcal{A}_{\varepsilon}(\PP)}\mathbb{E}_{(x,y)\sim\QQ,\theta\sim\mu}\left[l(\theta,(x,y))\right].
\label{eq:primal_pb}
\end{align}
In this game the classifier objective is to find the best distribution $\mu \in \mathcal{M}_1^+(\Theta)$ while the adversary is manipulating the data distribution. For the classifier, solving the infimum problem in Equation~\eqref{eq:primal_pb} simply amounts to solving the adversarial risk minimization problem -- Problem~\eqref{eq:advriskmin}, whether the classifier is randomized or not. Then, given a randomized classifier $\mu \in \mathcal{M}_1^+(\Theta)$, the goal of the attacker is to find a new data-set distribution $\mathbb{Q}$ in the set of adversarial distributions $\mathcal{A}_{\varepsilon}(\PP)$ that maximizes the risk of $\mu$. More formally, the adversary looks for $$\mathbb{Q} \in \argmaxB_{\QQ\in\mathcal{A}_{\varepsilon}(\PP)} \mathbb{E}_{(x,y)\sim\QQ,\theta\sim\mu}\left[ l(\theta,(x,y))\right]. $$
In the game theoretic terminology, $\mathbb{Q}$ is also called the best response of the attacker to the classifier $\theta$.

\begin{rmq}
Note that for a given classifier $\mu$ there always exists a ``deterministic'' best response, i.e. every single point $x$ is mapped to another single point $T(x)$. Let $T:\mathcal{X}\to\mathcal{X}$ be defined such that for all $x\in\mathcal{X}$, $l(T(x),y) = \sup_{x',~d(x,x')\leq\varepsilon} l(x',y)$. Thanks to~\citep[Proposition 7.50]{bertsekas2004stochastic}, $(T,id)$ is $\PP$-measurable. Then $(T,id)_\sharp \PP$ belongs to $\text{BR}(\mu)$. Therefore, $T$ is the optimal ``deterministic'' attack against the classifier $\mu$.
\end{rmq}



\subsection{Dual Formulation of the Game}

Every zero sum game has a dual formulation that allows for a deeper understanding of the framework. Here, from Proposition~\ref{prop:dro_adv}, we can define the dual problem of adversarial risk minimization for randomized classifiers. This dual problem also characterizes a two-player zero-sum game that writes as follows,
\begin{align}
    \sup_{\QQ\in\mathcal{A}_{\varepsilon}(\PP)}\inf_{\mu\in\mathcal{M}^1_+(\Theta)}\mathbb{E}_{(x,y)\sim\QQ,\theta\sim\mu}\left[l(\theta,(x,y))\right].
\label{eq:dual_pb}
\end{align}
In this dual game problem, the adversary plays first and seeks an adversarial distribution that has the highest possible risk when faced with an arbitrary classifier. This means that it has to select an adversarial perturbation for every input $x$, without seeing the classifier first. In this case, as pointed out by the motivating example in Section~\ref{sec:motiv-ex}, the attack can (and should) be randomized to ensure maximal harm against several classifiers. Then, given an adversarial distribution, the classifier objective is to find the best possible classifier on this distribution. Let us denote $\dualvalue^\varepsilon$ the value of the dual problem. Since the weak duality is always satisfied, we get
\begin{align}
\label{eq:weak_duality}
\dualvalue^\varepsilon\leq \valuerand^\varepsilon\leq \valuedet^\varepsilon.
\end{align}
Inequalities in Equation~\eqref{eq:weak_duality} mean that the lowest risk the classifier can get (regardless of the game we look at) is $\mathcal{D}^\varepsilon$. In particular, this means that the primal version of the game, \emph{i.e.} the adversarial risk minimization problem, will always have a value greater or equal to $\mathcal{D}^\varepsilon$. As we discussed in Section~\ref{sec:motiv-ex}, this lower bound may not be attained by a deterministic classifier. As we will demonstrate in the next section, optimizing over randomized classifiers allows to approach $\mathcal{D}^\varepsilon$ arbitrary closely.

\begin{rmq}
Note that, we can always define the dual problem when the classifier is deterministic, \begin{align*}
    \sup_{\QQ\in\mathcal{A}_{\varepsilon}(\PP)}\inf_{\theta\in\Theta}\mathbb{E}_{(x,y)\sim\QQ}\left[l(\theta,(x,y))\right].
\end{align*}
Furthermore, we can demonstrate that the dual problems for deterministic and randomized classifiers have the same value \footnote{See Appendix~\ref{sec:complements} for more details}; hence the inequalities in Equation~\eqref{eq:weak_duality}.
\end{rmq}

\subsection{Nash Equilibria for Randomized Strategies}

In the adversarial examples game, a Nash equilibrium is a couple $(\mu^*,\QQ^*)\in\mathcal{M}^1_+(\Theta)\times\mathcal{A}_\varepsilon(\PP)$ where both the classifier and the attacker have no incentive to deviate unilaterally from their strategies $\mu^*$ and $\QQ^*$. More formally, $(\mu^*,\QQ^*)$ is a Nash equilibrium of the adversarial examples game if $(\mu^*,\QQ^*)$ is a saddle point of the objective function $$(\mu,\QQ)\mapsto \mathbb{E}_{(x,y)\sim\QQ,\theta\sim\mu}\left[l(\theta,(x,y))\right].$$ Alternatively, we can say that $(\mu^*,\QQ^*)$ is a Nash equilibrium if and only if $\mu^*$ solves the adversarial risk minimization problem -- Problem~\eqref{eq:advriskmin}, $\QQ^*$ the dual problem -- Problem~\eqref{eq:duality}, and $\mathcal{D}^\varepsilon=\mathcal{V}_{rand}^\varepsilon$. In our problem, $\QQ^*$ always exists but it might not be the case for $\mu^*$. Then for any $\delta>0$, we say that  $(\mu_\delta,\QQ^*)$ is a $\delta$-approximate Nash equilibrium if $\QQ^*$ solves the dual problem and $\mu_\delta$ satisfies $\mathcal{D}^\varepsilon\geq\riskadv^\varepsilon(\mu_\delta)-\delta$. 





We now state our main result: the existence of approximate Nash equilibria in the adversarial examples game when both the classifier and the adversary can use randomized strategies. More precisely, we demonstrate that the duality gap between the adversary and the classifier problems is zero, which gives as a corollary the existence of Nash equilibria. 

\begin{thm}
\label{thm:duality-rand}
Let $\PP\in\mathcal{M}^1_+(\mathcal{X}\times\mathcal{Y})$. Let $\varepsilon>0$. Let $l:\Theta\times(\mathcal{X}\times\mathcal{Y})\to [0,\infty)$ satisfying Assumption~\ref{ass:loss}. 
Then strong duality always holds in the randomized  setting:
\begin{align}
 \label{eq:duality}\inf_{\mu\in \mathcal{M}^+_1(\Theta)} \max_{\QQ\in \mathcal{A}_{\varepsilon}(\PP)} \mathbb{E}_{\theta \sim \mu, (x,y)\sim\QQ }\left[l(\theta,(x,y))\right]\\
=
\nonumber\max_{\QQ\in \mathcal{A}_{\varepsilon}(\PP)}\inf_{\mu\in \mathcal{M}^+_1(\Theta)}  \mathbb{E}_{\theta \sim \mu, (x,y)\sim\QQ }\left[l(\theta,(x,y))\right]
\end{align}
The supremum is always attained. If $\Theta$ is a compact set, and for all $(x,y)\in\mathcal{X}\times\mathcal{Y}$, $l(\cdot,(x,y))$ is lower semi-continuous, the infimum is also attained.
\end{thm}

 \begin{corollary}
\label{cor:nash-eq}
Under Assumption~\ref{ass:loss}, for any $\delta>0$, there exists a $\delta$-approximate Nash-Equibilrium $(\mu_\delta,\QQ^*)$. Moreover, if the infimum is attained, there exists a Nash equilibrium $(\mu^*,\QQ^*)$ to the adversarial examples game.
\end{corollary}

Theorem~\ref{thm:duality-rand} shows that $\mathcal{D}^\varepsilon=\mathcal{V}_{rand}^\varepsilon$. From a game theoretic perspective, this means that the minimal adversarial risk for a randomized classifier against any attack (primal problem) is the same as the maximal risk an adversary can get by using an attack strategy that is oblivious to the classifier it faces (dual problem). This suggests that playing randomized strategies for the classifier could substantially  improve robustness to adversarial examples.
In the next section, we will design an algorithm that efficiently learn this classifier, we will get improve adversarial robustness over classical deterministic defenses.

\begin{rmq}
Theorem~\ref{thm:duality-rand} remains true if one replaces $\mathcal{A}_\varepsilon(\PP)$ with any other Wasserstein compact uncertainty sets (see~\cite{yue2020linear} for conditions of compactness). 
\end{rmq}

\section{Finding the Optimal Classifiers}
\label{sec:algo}

\subsection{An Entropic Regularization}

Let $\{(x_i,y_i)\}_{i=1}^N$ samples independently drawn from $\PP$ and denote  $\widehat{\mathbb
{P}}:=\frac{1}{N}\sum_{i=1}^N \delta_{(x_i,y_i)}$ the associated empirical distribution. 
One can show the adversarial empirical risk minimization can be casted as:
\begin{align*}
\widehat{\mathcal{R}}_{adv}^{\varepsilon,*}:=\inf_{\mu\in \mathcal{M}^+_1(\Theta)}\sum_{i=1}^N\sup_{\QQ_i\in\Gamma_{i,\varepsilon}}\mathbb{E}_{(x,y)\sim \QQ_i,\theta \sim \mu}\left[l(\theta,(x,y))\right]
\end{align*}
where $\Gamma_{i,\varepsilon}$ is defined as : 
\begin{align*}
    \Gamma_{i,\varepsilon}:=\Big\{\QQ_i\mid~\int d\QQ_i=\frac{1}{N},~\int c_{\varepsilon}((x_i,y_i),\cdot) d\QQ_i=0\Big\}.
\end{align*}
More details on this decomposition are given in Appendix~\ref{sec:complements}.
In the following, we regularize the above objective by adding an entropic term to each inner supremum problem. Let $\bm{\alpha}:=(\alpha_i)_{i=1}^N\in\mathbb{R}_+^N$ such that for all $i\in\{1,\dots,N\}$, and let us consider the following optimization problem:
\begin{equation*}
\begin{aligned}
\label{eq-legendre-KL}
\widehat{\mathcal{R}}_{adv,\bm{\alpha}}^{\varepsilon,*}:=\inf_{\mu\in \mathcal{M}^+_1(\Theta)}\sum_{i=1}^N&\sup_{\QQ_i\in\Gamma_{i,\varepsilon}}\mathbb{E}_{ \QQ_i, \mu}\left[l(\theta,(x,y))\right]\\
&-\alpha_i\text{KL}\left(\QQ_i\Big|\Big|\frac{1}{N}\mathbb{U}_{(x_i,y_i)}\right)
\end{aligned}
\end{equation*}
where $\mathbb{U}_{(x,y)}$ is an arbitrary distribution of support equal to:
\begin{align*}
    S_{(x,y)}^{(\varepsilon)}:=\Big\{(x',y'):~\text{s.t.}~c_{\varepsilon}((x,y),(x',y'))=0\Big\},
\end{align*}
and for all $\QQ,\mathbb{U}\in\mathcal{M}_{+}(\mathcal{X}\times\mathcal{Y})$,
\begin{align*}
\text{KL}(\QQ||\mathbb{U}):=  \left\{
    \begin{array}{lll}
        \int\log(\frac{d\QQ}{d\mathbb{U}})d\QQ+|\mathbb{U}| - |\QQ| &  \mbox{if } \QQ\ll \mathbb{U}\\
        +\infty & \mbox{otherwise.}
    \end{array}
\right.
\end{align*}
Note that when $\bm{\alpha}=0$, we recover the problem of interest $\widehat{\mathcal{R}}_{adv,\bm{\alpha}}^{\varepsilon,*}=\widehat{\mathcal{R}}_{adv}^{\varepsilon,*}$. Moreover, we show the regularized supremum tends to the standard supremum when $\bm{\alpha}\to 0$.
\begin{prop}
\label{prop:limit-eps}
For $\mu\in\mathcal{M}_{1}^{+}(\Theta)$, one has
\begin{align*}
    &\lim_{\alpha_i\rightarrow 0}\sup_{\QQ_i\in\Gamma_{i,\varepsilon}}\mathbb{E}_{\QQ_i,\mu}\left[l(\theta,(x,y))\right]-\alpha_i\text{KL}\left(\QQ\Big|\Big|\frac{1}{N}\mathbb{U}_{(x_i,y_i)}\right)\\
    &=\sup_{\QQ_i\in\Gamma_{i,\varepsilon}}\mathbb{E}_{(x,y)\sim \QQ_i,\theta \sim \mu}\left[l(\theta,(x,y))\right].
\end{align*}
\end{prop}
By adding an entropic term to the objective, we obtain an explicit formulation of the supremum involved in the sum: as soon as $\bm{\alpha}>0$ (which means that each $\alpha_i>0$), each sub-problem becomes just the Fenchel-Legendre transform of $\text{KL}(\cdot|\mathbb{U}_{(x_i,y_i)}/N)$ which has the following closed form:
\begin{align*}
 &\sup_{\QQ_i\in\Gamma_{i,\varepsilon}}\mathbb{E}_{\QQ_i, \mu}\left[l(\theta,(x,y))\right]-\alpha_i\text{KL}\left(\QQ_i||\frac{1}{N}\mathbb{U}_{(x_i,y_i)}\right)\\
 &=\frac{\alpha_i}{N}\log\left( \int_{\mathcal{X}\times\mathcal{Y}}\exp\left(\frac{\mathbb{E}_{\theta \sim \mu}\left[l(\theta,(x,y))\right]}{\alpha_i}\right)d\mathbb{U}_{(x_i,y_i)}\right).
\end{align*}
Finally, we end up with the following problem: 
\begin{align*}
  \inf_{\mu\in \mathcal{M}^+_1(\Theta)}\sum_{i=1}^N  \frac{\alpha_i}{N}\log\left( \int\exp\frac{\mathbb{E}_{ \mu}\left[l(\theta,(x,y))\right]}{\alpha_i}d\mathbb{U}_{(x_i,y_i)}\right).
\end{align*}
In order to solve the above problem, one needs to compute the integral involved in the objective. To do so, we estimate it by randomly sampling $m_i\geq 1$ samples $(u_1^{(i)},\dots,u_{m_i}^{(i)})\in(\mathcal{X}\times\mathcal{Y})^{m_i}$ from $\mathbb{U}_{(x_i,y_i)}$ for all $i\in\{1,\dots,N\}$ which leads to the following optimization problem
\begin{align}
\label{eq-obj-sample}
  \inf_{\mu\in \mathcal{M}^+_1(\Theta)}\sum_{i=1}^N  \frac{\alpha_i}{N}\log\left( \frac{1}{m_i}\sum_{j=1}^{m_i}\exp\frac{\mathbb{E}_{ \mu}\left[l(\theta,u_j^{(i)})\right]}{\alpha_i}\right)
\end{align}
denoted $\widehat{\mathcal{R}}_{adv,\bm{\alpha}}^{\varepsilon,\bm{m}}$ where $\bm{m}:=(m_i)_{i=1}^N$ in the following. Now we aim at controlling the error made with our approximations. We decompose the error into two terms
\begin{align*}
  |\widehat{\mathcal{R}}_{adv,\bm{\alpha}}^{\varepsilon,\bm{m}} - \widehat{\mathcal{R}}_{adv}^{\varepsilon,*}|
   \leq |\widehat{\mathcal{R}}_{adv,\bm{\alpha}}^{\varepsilon,*} - \widehat{\mathcal{R}}_{adv,\bm{\alpha}}^{\varepsilon,\bm{m}}| +|\widehat{\mathcal{R}}_{adv,\bm{\alpha}}^{\varepsilon,*} - \widehat{\mathcal{R}}_{adv}^{\varepsilon,*}|
\end{align*}
where the first one corresponds to the statistical error made by our estimation of the integral, and the second to the approximation error made by the entropic regularization of the objective. First, we show a control of the statistical error using Rademacher complexities~\cite{bartlett2002rademacher}. 
\begin{prop}
\label{prop:control-error-stat}
Let $m\geq 1$ and $\alpha>0$ and denote $\bm{\alpha}:=(\alpha,\dots,\alpha)\in\mathbb{R}^N$ and $\bm{m}:=(m,\dots,m)\in\mathbb{R}^N$.  \textcolor{black}{Then by denoting $\tilde{M}=\max(M,1)$}, we have with a probability of at least $1-\delta$
\begin{align*}
|\widehat{\mathcal{R}}_{adv,\bm{\alpha}}^{\varepsilon,*} - \widehat{\mathcal{R}}_{adv,\bm{\alpha}}^{\varepsilon,\bm{m}}|\leq& \frac{2e^{M/\alpha}}{N}\sum_{i=1}^N R_i + \color{black}{6\tilde{M}}\color{black}e^{M/\alpha}\sqrt{\frac{\log(\frac4\delta)}{2mN}}
\end{align*}
where $R_i:=\frac{1}{m}\mathbb{E}_{\bm{\sigma}}\left[\sup_{\theta\in\Theta}\sum_{j=1}^m \sigma_j l(\theta,u_j^{(i)})\right]$ and $\bm{\sigma}:=(\sigma_1,\dots,\sigma_m)$ with $\sigma_i$ i.i.d. sampled as $\mathbb{P}[\sigma_i=\pm1]=1/2$.
\end{prop}

We deduce from the above Proposition that in the particular case where $\Theta$ is finite such that $|\Theta|= L$, with probability of at least $1-\delta$
\begin{align*}
   |\widehat{\mathcal{R}}_{adv,\bm{\alpha}}^{\varepsilon,*} - \widehat{\mathcal{R}}_{adv,\bm{\alpha}}^{\varepsilon,\bm{m}}| \in \mathcal{O}\left(Me^{M/\alpha}\sqrt{\frac{\log(L)}{m}} \right).
\end{align*}
This case is of particular interest when one wants to learn the optimal mixture of some given classifiers in order to minimize the adversarial risk. In the following proposition, we control the approximation error made by adding an entropic term to the objective. 
\begin{prop}
\label{prop:control-error-approx}
Denote for $\beta>0$, $(x,y)\in\mathcal{X}\times\mathcal{Y}$ and $\mu\in\mathcal{M}_{1}^{+}(\Theta)$,
    $A_{\beta,\mu}^{\left(  x,y\right)}:=\{u|\sup_{v\in  S_{(x,y)}^{(\varepsilon)}}\mathbb{E}_{\mu}[l(\theta,v)]\leq \mathbb{E}_{\mu}[l(\theta,u)]+\beta\}$. 
    If there exists $C_{\beta}$ such that for all $(x,y)\in\mathcal{X}\times\mathcal{Y}$ and $\mu\in\mathcal{M}_{1}^{+}(\Theta)$, $\mathbb{U}_{(x,y)}\left(A_{\beta,\mu}^{\left(  x,y\right)}\right)\geq C_\beta$ then we have
\begin{align*}
   |\widehat{\mathcal{R}}_{adv,\bm{\alpha}}^{\varepsilon,*} - \widehat{\mathcal{R}}_{adv}^{\varepsilon,*}|\leq 2\alpha |\log(C_\beta)| + \beta.
\end{align*}
\end{prop}
The assumption made in the above Proposition states that for any given random classifier $\mu$, and any given point $(x,y)$, the set of $\beta$-optimal attacks at this point has at least a certain amount of mass depending on the $\beta$ chosen. This assumption is always met when $\beta$ is sufficiently large. However in order to obtain a tight control of the error, a trade-off exists between $\beta$ and the smallest amount of mass $C_{\beta}$ of $\beta$-optimal attacks.

Now that we have shown that solving~\eqref{eq-obj-sample} allows to obtain an approximation of the true solution $\widehat{\mathcal{R}}_{adv}^{\varepsilon,*}$, we next aim at deriving an algorithm to compute it. 

\subsection{Proposed Algorithms}
\label{sec:proposed-algo}
From now on, we focus on finite class of classifiers. Let $\Theta = \{\theta_1,\dots,\theta_L\}$, we aim to learn the optimal mixture of classifiers in this case. The adversarial empirical risk  is therefore defined as:
\begin{align*}
    \widehat{\mathcal{R}}_{adv}^\varepsilon(\bm{\lambda})= \sum_{i=1}^N\sup_{\QQ_i\in\Gamma_{i,\varepsilon}}\mathbb{E}_{(x,y)\sim \QQ_i}\left[\sum_{k=1}^L \lambda_k l(\theta_k,(x,y))\right]
\end{align*}
for $\bm{\lambda}\in\Delta_L: = \{\bm{\lambda}\in\mathbb{R}_+^L~\mathrm{s.t.}~\sum_{i=1}^L\lambda_i=1\}$, the probability simplex of $\mathbb{R}^L$. One can notice that $ \widehat{\mathcal{R}}_{adv}^\varepsilon(\cdot)$ is a continuous convex function, hence $\min_{\bm{\lambda}\in\Delta_L}\riskadv^\varepsilon(\bm{\lambda})$ is attained for a certain $\bm{\lambda}^*$. Then there exists a non-approximate Nash equilibrium $(\bm{\lambda}^*,\QQ^*)$ in the adversarial game when $\Theta$ is finite. Here, we present two algorithms to learn the optimal mixture of the adversarial risk minimization problem.
\begin{algorithm}[H]
\small
\SetAlgoLined
 $\bm{\lambda}_0 = \frac{\mathbf{1}_L}{L}; T;~\textcolor{black}{\eta=\frac{2}{M\sqrt{LT}}}$\\
 \For{$t=1,\dots,T$}{

  $\Tilde{\QQ}$ s.t. $\exists\QQ^*\in\mathcal{A}_\varepsilon(\PP)$ best response to \textcolor{black}{$\bm{\lambda}_{t-1}$} and for all $k\in[L]$, $\lvert\mathbb{E}_{\Tilde{\QQ}}(l(\theta_k,(x,y)))-\mathbb{E}_{\QQ^*}(l(\theta_k,(x,y))) \rvert\leq\delta$\\
  $\bm{g}_t=\left(\mathbb{E}_{\Tilde{\QQ}}(l(\theta_1,(x,y)),\dots,\mathbb{E}_{\Tilde{\QQ}}(l(\theta_L,(x,y))\right)^T$\\
  $\bm{\lambda}_t = \Pi_{\Delta_L}\left(\bm{\lambda}_{t-1}-\eta \bm{g}_t\right)$
  }
 \caption{Oracle Algorithm}
 \label{algo:duchi}
\end{algorithm}
\begin{figure*}[!h]
    \centering
    \includegraphics[width=0.32\textwidth]{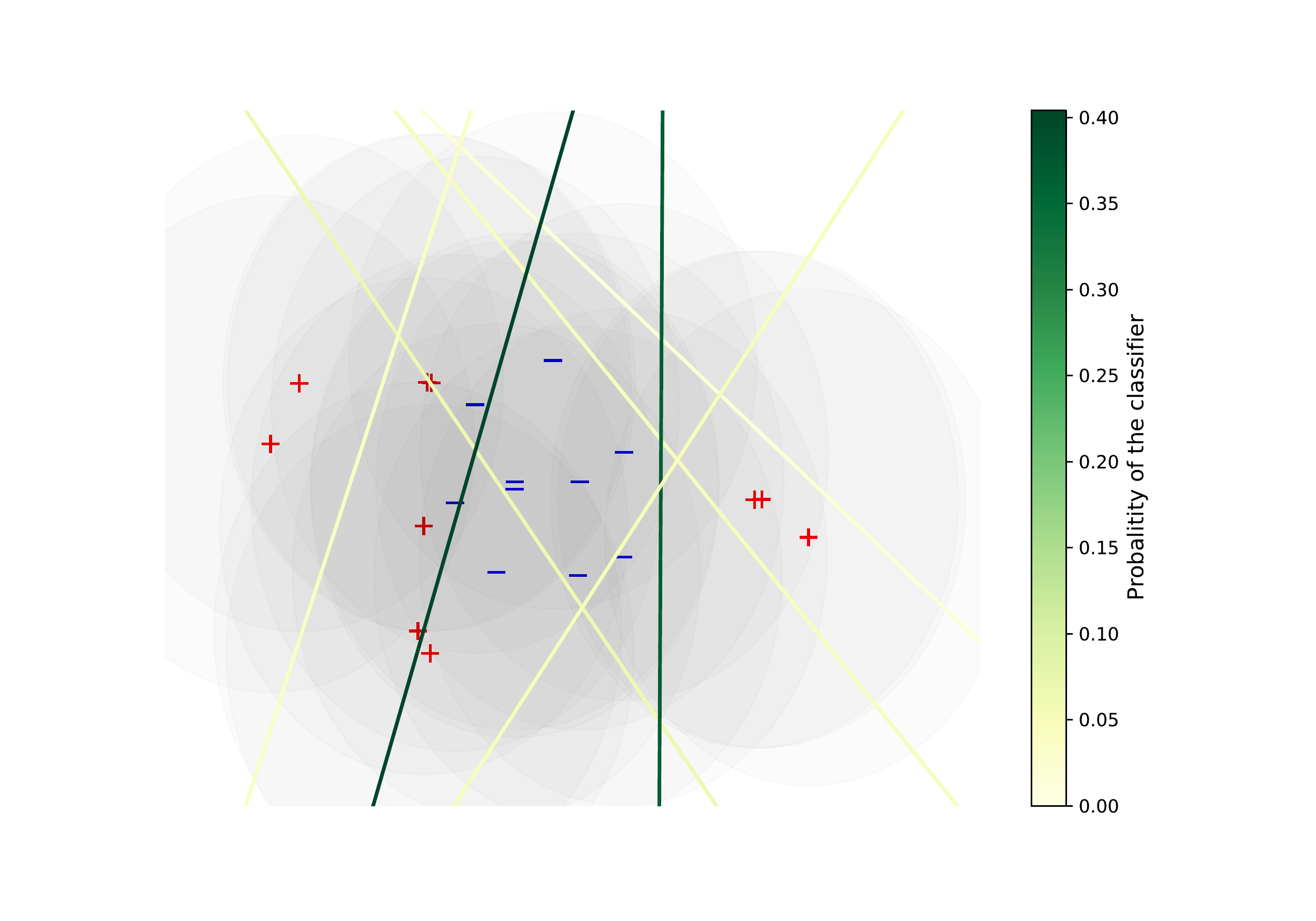}  \includegraphics[width=0.32\textwidth]{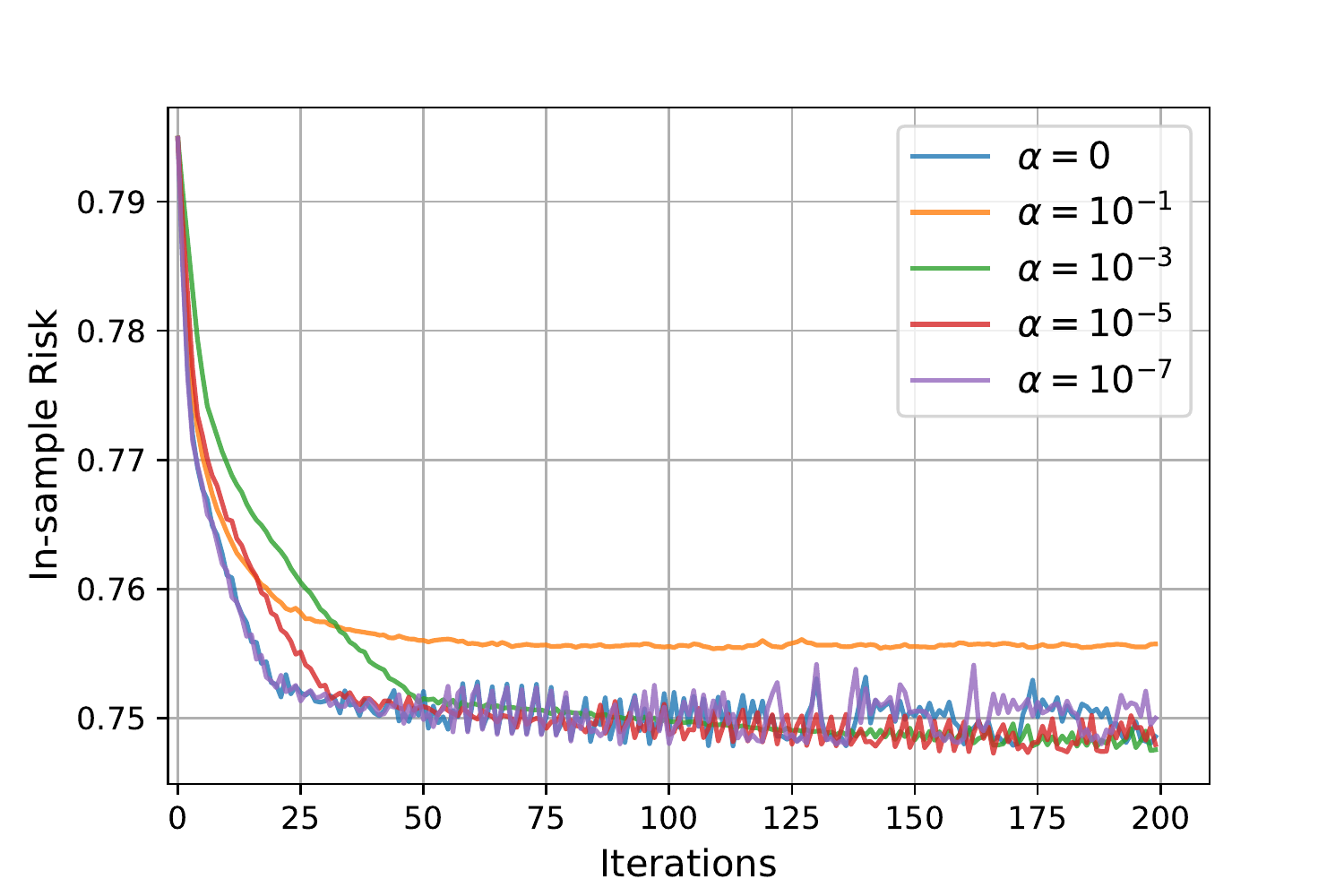}     \includegraphics[width=0.32\textwidth]{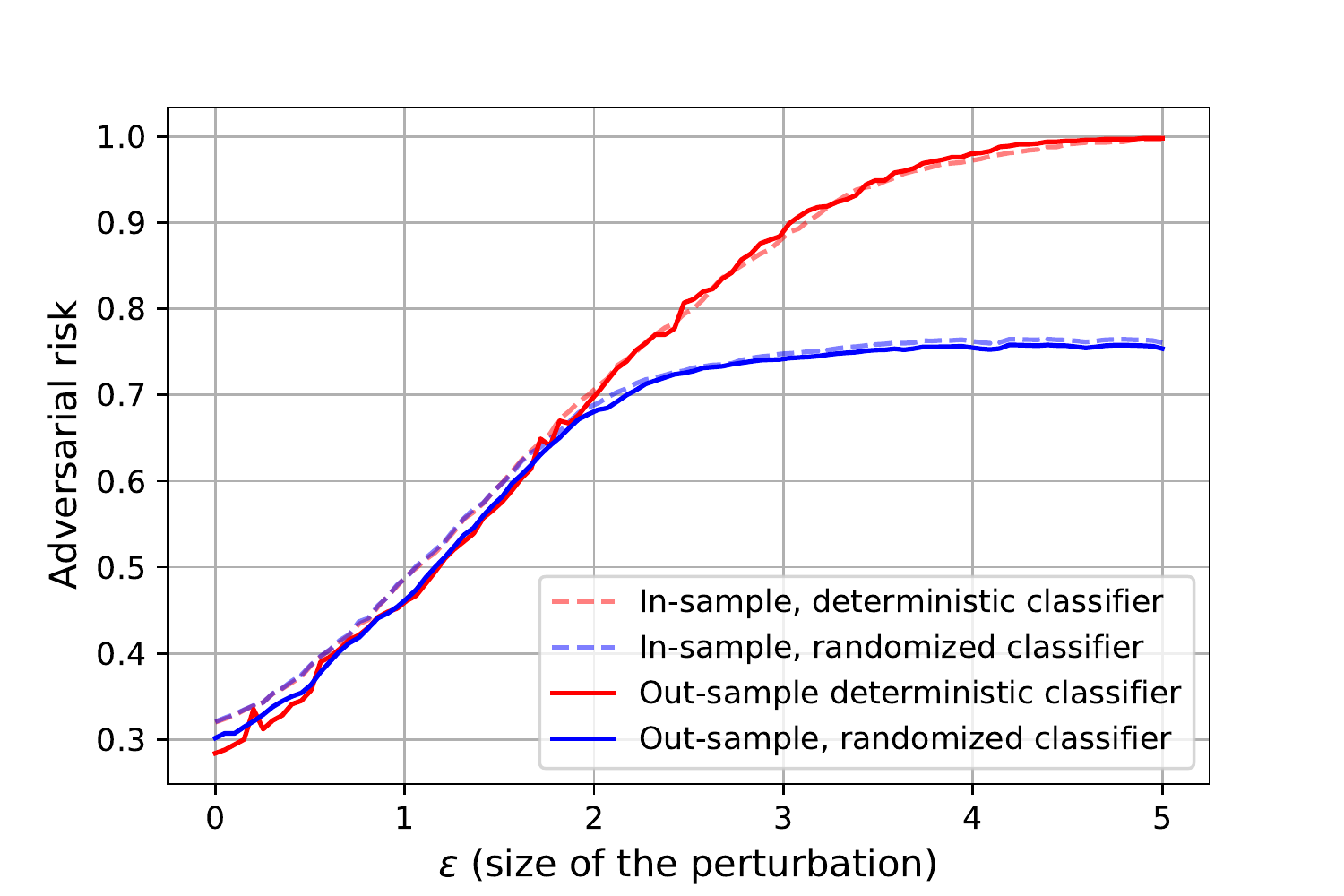}
    \caption{On left, $40$ data samples with their set of possible attacks represented in shadow and the optimal randomized classifier, with a color gradient representing the probability of the classifier. \textcolor{black}{In the middle}, convergence of the oracle ($\alpha=0$) and regularized algorithm for different values of regularization parameters. On right, in-sample and out-sample risk for randomized and deterministic minimum risk in function of the perturbation size $\varepsilon$. In the latter case, the randomized classifier is optimized with oracle Algorithm~\ref{algo:duchi}.}
    \label{fig:toy_example}
\end{figure*}
\textbf{A First Oracle Algorithm.} The first algorithm we present is inspired from~\cite{sinha2017certifying} and the convergence of projected sub-gradient methods~\cite{boyd2003subgradient}. The computation of the inner supremum problem is usually NP-hard~\footnote{See Appendix~\ref{sec:complements} for details.}, but one may assume the existence of an approximate oracle to this supremum. The algorithm is presented in Algorithm~\ref{algo:duchi}. We get the following guarantee for this algorithm. 
\begin{prop}
\label{prop:algo-oracle}
Let $l:\Theta\times(\mathcal{X}\times\mathcal{Y})\to [0,\infty)$ satisfying Assumption~\ref{ass:loss}. Then, Algorithm~\ref{algo:duchi} satisfies:  
\begin{align*}
    \min_{t\in[T]} \widehat{\mathcal{R}}_{adv}^{\varepsilon}(\bm{\lambda}_t)-\widehat{\mathcal{R}}_{adv}^{\varepsilon,*}\leq2\delta+\textcolor{black}{ \frac{2M\sqrt{L}}{\sqrt{T}}}
\end{align*}
\end{prop}
The main drawback of the above algorithm is that one needs to have access to an oracle to guarantee the convergence of the proposed algorithm. In the following we present its regularized version in order to approximate the solution and propose a simple algorithm to solve it.

\textbf{An Entropic Relaxation.} Adding an entropic term to the objective allows to have a simple reformulation of the problem, as follows:
\begin{align*}
  \inf_{\bm{\lambda}\in \Delta_L}\sum_{i=1}^N  \frac{\varepsilon_i}{N}\log\left( \frac{1}{m_i}\sum_{j=1}^{m_i}\exp\left(\frac{\sum_{k=1}^L \lambda_kl(\theta_k,u_j^{(i)})}{\varepsilon_i}\right)\right)
\end{align*}
Note that in $\bm{\lambda}$, the objective is convex and smooth. One can  apply the accelerated PGD~\citep{beck2009fast,tseng2008accelerated} which enjoys an optimal convergence rate for first order methods of $\mathcal{O}(T^{-2})$ for $T$ iterations. 

\section{Experiments}

\subsection{Synthetic Dataset}

To illustrate our theoretical findings, we start by testing our learning algorithm on the following synthetic two-dimensional problem. Let us consider the distribution $\PP$ defined as  $\PP\left(Y =\pm 1\right)=1/2$, $\PP\left(X\mid Y=-1\right) = \mathcal{N}(0,I_2)$ and $\PP\left(X \mid Y=1\right) = \frac12\left[\mathcal{N}((-3,0),I_2)+\mathcal{N}((3,0),I_2) \right]$.
We sample $1000$ training points from this distribution and randomly generate $10$ linear classifiers that achieves a standard training risk lower than $0.4$. To simulate an adversary with budget $\varepsilon$ in $\ell_2$ norm, we proceed as follows. For every sample $(x,y)\sim \PP$ we generate $1000$ points uniformly at random in the ball of radius $\varepsilon$ and select the one maximizing the risk for the $0/1$ loss. Figure~\ref{fig:toy_example} (left) illustrates the type of mixture we get after convergence of our algorithms. Note that in this toy problem, we are likely to find the optimal adversary with this sampling strategy if we sample enough attack points. 

To evaluate the convergence of our algorithms, we compute the adversarial risk of our mixture for each iteration of both the oracle and regularized algorithms. Figure~\ref{fig:toy_example} illustrates the convergence of the algorithms w.r.t the regularization parameter. We observe that the risk for both algorithms converge. Moreover, they converge towards the oracle minimizer when the regularization parameter $\alpha$ goes to $0$.

Finally, to demonstrate the improvement randomized techniques offer against deterministic defenses, we plot in Figure~\ref{fig:toy_example} (right) the minimum adversarial risk for both randomized and deterministic classifiers w.r.t. $\varepsilon$. The adversarial risk is strictly better for randomized classifier whenever the adversarial budget $\varepsilon$ is bigger than $2$. This illustration validates our analysis of Theorem~\ref{thm:duality-rand}, and motivates a in depth study of a more challenging framework, namely image classification with neural networks.


\subsection{CIFAR-10 Dataset}

\begin{figure*}[!ht]
\begin{center}

\vskip 0.15in
 \begin{minipage}[ht!]{0.39\textwidth}
 \begin{scriptsize}
\begin{tabular}{c|c|ccc} 
\textbf{ Models} & \textbf{Acc. }&\textbf{$\textrm{APGD}_\textrm{CE}$}& \textbf{$\textrm{APGD}_\textrm{DLR}$} & \textbf{Rob. Acc.} \\ \hline
 1 & $81.9\%$ &	$47.6\%$ & $47.7\%$ & $45.6\%$ \\ 
 2 & $81.9\%$ & $49.0\%$ & ${49.6\%}$ & ${47.0\%}$\\ 
  3 & ${81.7\%}$& ${49.0\%}$ & $49.3\%$ & ${46.9\%}$\\
    4 & $\bm{82.6\%}$& $\bm{49.7\%}$ & $\bm{49.8}\%$ & $\bm{47.2\%}$\\

\end{tabular}
\end{scriptsize}
  \end{minipage}\begin{minipage}[!ht]{0.61\textwidth}
\includegraphics[width=0.49\textwidth]{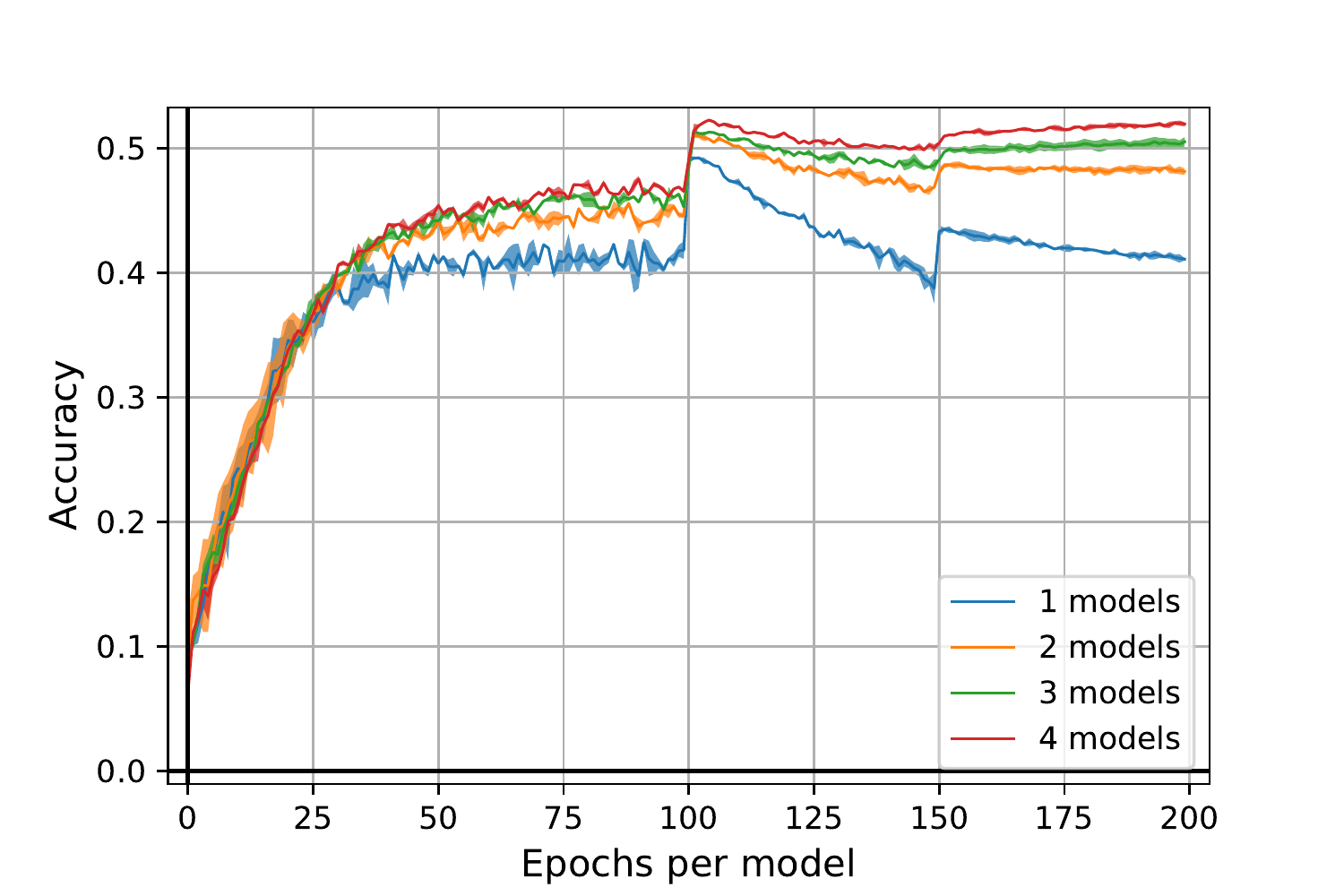}\includegraphics[width=0.49\textwidth]{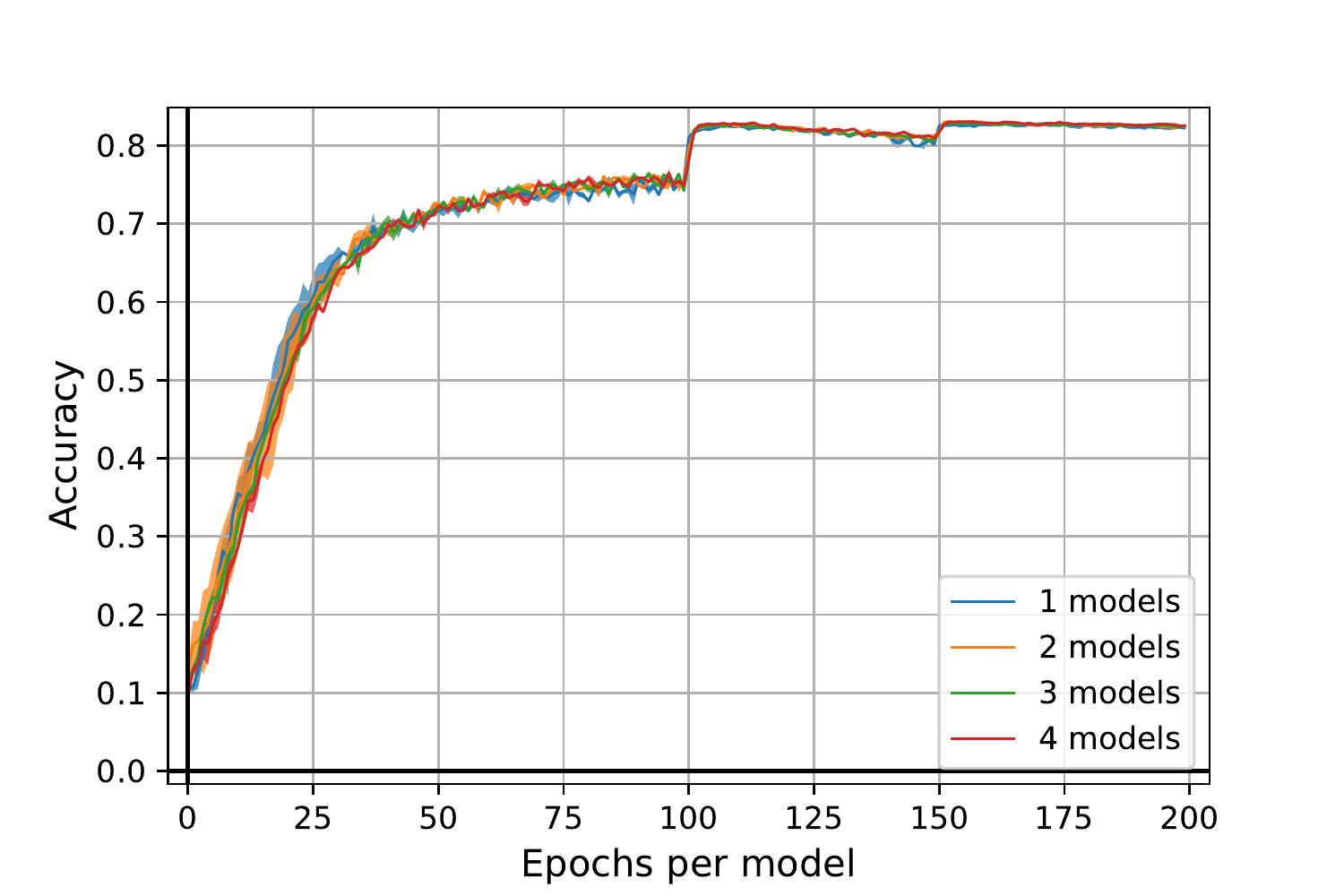} 
  \end{minipage}
  
\caption{\textcolor{black}{On left}: Comparison of our algorithm with a standard adversarial training (one model). We reported the results for the model with the best robust accuracy obtained over two independent runs because adversarial training might be unstable. Standard and Robust accuracy (respectively in the middle and on right) on CIFAR-10 test images in function of the number of epochs per classifier with $1$ to $3$ ResNet18 models. The performed attack is PGD with $20$ iterations and $\varepsilon=8/255$.}
\label{fig:results_cifar}
\end{center}
\end{figure*}

Adversarial examples are known to be easily transferrable from one model to another~\cite{tramer2017space}. To counter this and support our theoretical claims, we propose an heuristic algorithm (see Algorithm~\ref{algo:heuristic}) to train a robust mixture of $L$ classifiers. We alternatively train these classifiers with adversarial examples against the current mixture and update the probabilities of the mixture according to the algorithms we proposed in Section~\ref{sec:proposed-algo}. More details on the heuristic algorithm are available in Appendix~\ref{sec:additional-xp}. 
\begin{algorithm}[h!]
\small
\SetAlgoLined
$L$: number of models, $T$: number of iterations,\\
$T_\theta$: number of updates for the models $\bm{\theta}$,\\
$T_\lambda$: number of updates for the mixture $\bm{\lambda}$,\\ $\bm{\lambda}_0=(\lambda_0^1,\dots\lambda_0^L),~\bm{\theta}_0=(\theta_0^1,\dots\theta_0^L)$\\
 \For{$t=1,\dots,T$}{
 Let $B_t$ be a batch of data.\\
\eIf{$t \mod (T_\theta L+1)\neq 0$}{
$k$ sampled uniformly in $\{1,\dots,L\}$\\
$\Tilde{B}_t\leftarrow$ Attack of images in $B_t$ for the  model $(\bm{\lambda}_t,\bm{\theta}_t)$\\
$\theta^t_k\leftarrow$ Update $\theta^{t-1}_k$ with $\Tilde{B}_t$ for fixed $\bm{\lambda}_t$ with a SGD step}{
$\bm{\lambda}_t\leftarrow$Update $\bm{\lambda}_{t-1}$ on $B_t$ for fixed $\bm{\theta}_t$
with oracle or regularized algorithm with $T_\lambda$ iterations.
}
  }
 \caption{Adversarial Training for Mixtures}
 
 \label{algo:heuristic}
\end{algorithm}

\textbf{Experimental Setup.} To evaluate the performance of Algorithm~\ref{algo:heuristic}, we trained from $1$ to $4$ ResNet18~\citep{He_2016_CVPR} models on $200$ epochs per model\footnote{$L\times200$ epochs in total, where $L$ is the number of models.}. We study the robustness with regards to $\ell_\infty$ norm and fixed adversarial budget $\varepsilon=8/255$. The attack we used in the inner maximization of the training is an adapted (adaptative) version of PGD for mixtures of classifiers with $10$ steps. Note that for one single model, Algorithm~\ref{algo:heuristic} exactly corresponds to adversarial training~\cite{madry2017towards}. For each of our setups, we made two independent runs and select the  best one. The training time of our algorithm is around four times longer than a standard Adversarial Training (with PGD 10 iter.) with two models, eight times with three models and twelve times with four models. We trained our models with a batch of size  $1024$ on $8$ Nvidia V100 GPUs. We give more details on implementation in Appendix~\ref{sec:additional-xp}. 

\textbf{Evaluation Protocol.} At each epoch, we evaluate the current mixture on test data against PGD attack  with $20$ iterations. To select our model and avoid overfitting~\cite{rice2020overfitting}, we kept the most robust against this PGD attack.
To make a final evaluation of our mixture of models, we used an adapted version of $\textrm{AutoPGD}$ untargeted attacks~\citep{croce2020reliable} for randomized classifiers with both Cross-Entropy (CE) and Difference of Logits Ratio (DLR) loss. For both attacks, we made $100$ iterations and $5$ restarts.

\textbf{Results.} The results are presented in Figure~\ref{fig:results_cifar}. We remark our algorithm outperforms a standard adversarial training in all the cases by more $1\%$, without additional loss of standard accuracy as it is attested by the left figure. Moreover, it seems our algorithm, by adding more and more models, reduces the overfitting of adversarial training. So far, experiments are computationally very costful and it is difficult to raise precise conclusions. Further, hyperparameter tuning ~\cite{gowal2020uncovering} such as architecture, unlabeled data~\cite{carmon2019unlabeled}, activation function, or the use of TRADES~\cite{zhang2019theoretically} may still increase the results.

\section{Related Work and Discussions}
\label{sec:rw}

\textbf{Distributionally Robust Optimization.}
Several recent works~\cite{sinha2017certifying,NIPS2018_7534,tu2018theoretical} studied the problem of adversarial examples through the scope of distributionally robust optimization. In these frameworks, the set of adversarial distributions is defined using an $\ell_p$ Wasserstein ball (the adversary is allowed to have an \emph{average} perturbation of at most $\varepsilon$ in $\ell_p$ norm). This however does not match the usual adversarial attack problem, where the adversary cannot move any point by more than $\varepsilon$. In the present work, we introduce a cost  function allowing us to cast the adversarial example problem as a DRO one, without changing the adversary constraints.

\textbf{Optimal Transport (OT).} 
\cite{NIPS2019_8968} and ~\cite{pydi2019adversarial} investigated classifier-agnostic lower bounds on the adversarial risk of any deterministic classifier using OT. These works  only evaluate lower bounds on the primal deterministic formulation of the problem, while we study the existence of mixed Nash equilibria.  Note that~\cite{pydi2019adversarial} started to investigate a way to formalize the adversary using Markov kernels, but did not investigate the impact of randomized strategies on the game. We extended this work by rigorously reformulating the adversarial risk as a linear optimization problem over distributions and we study this problem from a game theoretic point of view.


\textbf{Game Theory.} Adversarial examples have been studied under  the notions of Stackelberg game in~\cite{10.1145/2020408.2020495}, and  zero-sum game in~\cite{7533509,DBLP:journals/corr/abs-1906-02816,bose2021adversarial}. These works considered restricted settings (convex loss, parametric adversaries, etc.) that do not  comply with the nature of the problem. Indeed, we prove in Appendix~\ref{prv:duality-rand} that when the loss is convex and the set $\Theta$ is convex, the duality gap is zero for deterministic classifiers.
However, it has been proven that no convex loss can be a good surrogate for the $0/1$ loss in the adversarial setting~\citep{pmlr-v125-bao20a,pmlr-v97-cranko19a}, narrowing the scope of this result. If one can  show that for sufficiently separated conditional distributions, an optimal deterministic classifier always exists (see Appendix~\ref{sec:complements} for a clear statement), necessary and sufficient conditions for the need of randomization are still to be established.  ~\cite{pinot2020randomization} studied partly this question for regularized deterministic adversaries, leaving the general setting of  randomized adversaries and  mixed equilibria unanswered, which is the very scope of this paper.

\newpage
\bibliography{example_paper}

\begin{thebibliography}{10}

\bibitem{pmlr-v125-bao20a}
H.~Bao, C.~Scott, and M.~Sugiyama.
\newblock Calibrated surrogate losses for adversarially robust classification.
\newblock In J.~Abernethy and S.~Agarwal, editors, {\em Proceedings of Thirty
  Third Conference on Learning Theory}, volume 125 of {\em Proceedings of
  Machine Learning Research}, pages 408--451. PMLR, 09--12 Jul 2020.

\bibitem{bartlett2002rademacher}
P.~L. Bartlett and S.~Mendelson.
\newblock Rademacher and gaussian complexities: Risk bounds and structural
  results.
\newblock {\em Journal of Machine Learning Research}, 3:463--482, 2002.

\bibitem{beck2009fast}
A.~Beck and M.~Teboulle.
\newblock A fast iterative shrinkage-thresholding algorithm for linear inverse
  problems.
\newblock {\em SIAM journal on imaging sciences}, 2(1):183--202, 2009.

\bibitem{bertsekas2004stochastic}
D.~P. Bertsekas and S.~Shreve.
\newblock {\em Stochastic optimal control: the discrete-time case}.
\newblock 2004.

\bibitem{NIPS2019_8968}
A.~N. Bhagoji, D.~Cullina, and P.~Mittal.
\newblock Lower bounds on adversarial robustness from optimal transport.
\newblock In {\em Advances in Neural Information Processing Systems 32}, pages
  7496--7508. Curran Associates, Inc., 2019.

\bibitem{biggio2013evasion}
B.~Biggio, I.~Corona, D.~Maiorca, B.~Nelson, N.~{\v{S}}rndi{\'c}, P.~Laskov,
  G.~Giacinto, and F.~Roli.
\newblock Evasion attacks against machine learning at test time.
\newblock In {\em Joint European conference on machine learning and knowledge
  discovery in databases}, pages 387--402. Springer, 2013.

\bibitem{blanchet2019quantifying}
J.~Blanchet and K.~Murthy.
\newblock Quantifying distributional model risk via optimal transport.
\newblock {\em Mathematics of Operations Research}, 44(2):565--600, 2019.

\bibitem{bose2021adversarial}
A.~J. Bose, G.~Gidel, H.~Berard, A.~Cianflone, P.~Vincent, S.~Lacoste-Julien,
  and W.~L. Hamilton.
\newblock Adversarial example games, 2021.

\bibitem{boyd2003subgradient}
S.~Boyd.
\newblock Subgradient methods.
\newblock 2003.

\bibitem{10.1145/2020408.2020495}
M.~Br\"{u}ckner and T.~Scheffer.
\newblock Stackelberg games for adversarial prediction problems.
\newblock In {\em Proceedings of the 17th ACM SIGKDD International Conference
  on Knowledge Discovery and Data Mining}, KDD ’11, page 547–555, New York,
  NY, USA, 2011. Association for Computing Machinery.

\bibitem{carlini2017adversarial}
N.~Carlini and D.~Wagner.
\newblock Adversarial examples are not easily detected: Bypassing ten detection
  methods.
\newblock In {\em Proceedings of the 10th ACM Workshop on Artificial
  Intelligence and Security}, pages 3--14, 2017.

\bibitem{carmon2019unlabeled}
Y.~Carmon, A.~Raghunathan, L.~Schmidt, P.~Liang, and J.~C. Duchi.
\newblock Unlabeled data improves adversarial robustness.
\newblock {\em arXiv preprint arXiv:1905.13736}, 2019.

\bibitem{KolterRandomizedSmoothing}
J.~M. Cohen, E.~Rosenfeld, and J.~Z. Kolter.
\newblock Certified adversarial robustness via randomized smoothing.
\newblock {\em arXiv preprint arXiv:1902.02918}.

\bibitem{pmlr-v97-cranko19a}
Z.~Cranko, A.~Menon, R.~Nock, C.~S. Ong, Z.~Shi, and C.~Walder.
\newblock Monge blunts bayes: Hardness results for adversarial training.
\newblock In K.~Chaudhuri and R.~Salakhutdinov, editors, {\em Proceedings of
  the 36th International Conference on Machine Learning}, volume~97 of {\em
  Proceedings of Machine Learning Research}, pages 1406--1415. PMLR, 09--15 Jun
  2019.

\bibitem{croce2020reliable}
F.~Croce and M.~Hein.
\newblock Reliable evaluation of adversarial robustness with an ensemble of
  diverse parameter-free attacks.
\newblock In {\em International Conference on Machine Learning}, 2020.

\bibitem{cuturi2013sinkhorn}
M.~Cuturi.
\newblock Sinkhorn distances: Lightspeed computation of optimal transport.
\newblock {\em Advances in neural information processing systems},
  26:2292--2300, 2013.

\bibitem{pruningDefenseICLR2018}
G.~S. Dhillon, K.~Azizzadenesheli, J.~D. Bernstein, J.~Kossaifi, A.~Khanna,
  Z.~C. Lipton, and A.~Anandkumar.
\newblock Stochastic activation pruning for robust adversarial defense.
\newblock In {\em International Conference on Learning Representations}, 2018.

\bibitem{goodfellow2014explaining}
I.~Goodfellow, J.~Shlens, and C.~Szegedy.
\newblock Explaining and harnessing adversarial examples.
\newblock In {\em International Conference on Learning Representations}, 2015.

\bibitem{gowal2020uncovering}
S.~Gowal, C.~Qin, J.~Uesato, T.~Mann, and P.~Kohli.
\newblock Uncovering the limits of adversarial training against norm-bounded
  adversarial examples.
\newblock {\em arXiv preprint arXiv:2010.03593}, 2020.

\bibitem{He_2016_CVPR}
K.~He, X.~Zhang, S.~Ren, and J.~Sun.
\newblock Deep residual learning for image recognition.
\newblock In {\em The IEEE Conference on Computer Vision and Pattern
  Recognition (CVPR)}, 2016.

\bibitem{krizhevsky2009learning}
A.~Krizhevsky and G.~Hinton.
\newblock Learning multiple layers of features from tiny images.
\newblock Technical report, Citeseer, 2009.

\bibitem{kurakin2016adversarial}
A.~Kurakin, I.~Goodfellow, and S.~Bengio.
\newblock Adversarial examples in the physical world.
\newblock {\em arXiv preprint arXiv:1607.02533}, 2016.

\bibitem{NIPS2018_7534}
J.~Lee and M.~Raginsky.
\newblock Minimax statistical learning with wasserstein distances.
\newblock In {\em Advances in Neural Information Processing Systems 31}, pages
  2687--2696. Curran Associates, Inc., 2018.

\bibitem{madry2017towards}
A.~Madry, A.~Makelov, L.~Schmidt, D.~Tsipras, and A.~Vladu.
\newblock Towards deep learning models resistant to adversarial attacks.
\newblock In {\em International Conference on Learning Representations}, 2018.

\bibitem{moosavi2019robustness}
S.-M. Moosavi-Dezfooli, A.~Fawzi, J.~Uesato, and P.~Frossard.
\newblock Robustness via curvature regularization, and vice versa.
\newblock In {\em Proceedings of the IEEE Conference on Computer Vision and
  Pattern Recognition}, pages 9078--9086, 2019.

\bibitem{DBLP:journals/corr/abs-1906-02816}
J.~C. Perdomo and Y.~Singer.
\newblock Robust attacks against multiple classifiers.
\newblock {\em arXiv preprint arXiv:1906.02816}, 2019.

\bibitem{pinot2020randomization}
R.~Pinot, R.~Ettedgui, G.~Rizk, Y.~Chevaleyre, and J.~Atif.
\newblock Randomization matters. how to defend against strong adversarial
  attacks.
\newblock {\em International Conference on Machine Learning}, 2020.

\bibitem{pinot2019theoretical}
R.~Pinot, L.~Meunier, A.~Araujo, H.~Kashima, F.~Yger, C.~Gouy-Pailler, and
  J.~Atif.
\newblock Theoretical evidence for adversarial robustness through
  randomization.
\newblock In {\em Advances in Neural Information Processing Systems}, pages
  11838--11848, 2019.

\bibitem{pydi2019adversarial}
M.~S. Pydi and V.~Jog.
\newblock Adversarial risk via optimal transport and optimal couplings.
\newblock In {\em International Conference on Machine Learning}. 2020.

\bibitem{rice2020overfitting}
L.~Rice, E.~Wong, and Z.~Kolter.
\newblock Overfitting in adversarially robust deep learning.
\newblock In {\em International Conference on Machine Learning}, pages
  8093--8104. PMLR, 2020.

\bibitem{7533509}
S.~{Rota Bulò}, B.~{Biggio}, I.~{Pillai}, M.~{Pelillo}, and F.~{Roli}.
\newblock Randomized prediction games for adversarial machine learning.
\newblock {\em IEEE Transactions on Neural Networks and Learning Systems},
  28(11):2466--2478, 2017.

\bibitem{shalev2014understanding}
S.~Shalev-Shwartz and S.~Ben-David.
\newblock {\em Understanding machine learning: From theory to algorithms}.
\newblock Cambridge university press, 2014.

\bibitem{sinha2017certifying}
A.~Sinha, H.~Namkoong, and J.~Duchi.
\newblock Certifying some distributional robustness with principled adversarial
  training.
\newblock {\em arXiv preprint arXiv:1710.10571}, 2017.

\bibitem{Szegedy2013IntriguingPO}
C.~Szegedy, W.~Zaremba, I.~Sutskever, J.~Bruna, D.~Erhan, I.~Goodfellow, and
  R.~Fergus.
\newblock Intriguing properties of neural networks.
\newblock In {\em International Conference on Learning Representations}, 2014.

\bibitem{tramer2017space}
F.~Tram{\`e}r, N.~Papernot, I.~Goodfellow, D.~Boneh, and P.~McDaniel.
\newblock The space of transferable adversarial examples.
\newblock {\em arXiv preprint arXiv:1704.03453}, 2017.

\bibitem{tseng2008accelerated}
P.~Tseng.
\newblock On accelerated proximal gradient methods for convex-concave
  optimization.
\newblock {\em submitted to SIAM Journal on Optimization}, 1, 2008.

\bibitem{tu2018theoretical}
Z.~Tu, J.~Zhang, and D.~Tao.
\newblock Theoretical analysis of adversarial learning: A minimax approach.
\newblock {\em arXiv preprint arXiv:1811.05232}, 2018.

\bibitem{villani2003topics}
C.~Villani.
\newblock {\em Topics in optimal transportation}.
\newblock Number~58. American Mathematical Soc., 2003.

\bibitem{NIPS2019_8443}
B.~Wang, Z.~Shi, and S.~Osher.
\newblock Resnets ensemble via the feynman-kac formalism to improve natural and
  robust accuracies.
\newblock In {\em Advances in Neural Information Processing Systems 32}, pages
  1655--1665. Curran Associates, Inc., 2019.

\bibitem{Xie2017MitigatingAE}
C.~Xie, J.~Wang, Z.~Zhang, Z.~Ren, and A.~Yuille.
\newblock Mitigating adversarial effects through randomization.
\newblock In {\em International Conference on Learning Representations}, 2018.

\bibitem{yue2020linear}
M.-C. Yue, D.~Kuhn, and W.~Wiesemann.
\newblock On linear optimization over wasserstein balls.
\newblock {\em arXiv preprint arXiv:2004.07162}, 2020.

\bibitem{ZagoruykoK16}
S.~Zagoruyko and N.~Komodakis.
\newblock Wide residual networks.
\newblock In {\em Proceedings of the British Machine Vision Conference (BMVC)},
  pages 87.1--87.12. BMVA Press, 2016.

\bibitem{zhang2019theoretically}
H.~Zhang, Y.~Yu, J.~Jiao, E.~P. Xing, L.~E. Ghaoui, and M.~I. Jordan.
\newblock Theoretically principled trade-off between robustness and accuracy.
\newblock {\em International conference on Machine Learning}, 2019.

\end{thebibliography}
\bibliographystyle{abbrv}
\newpage
\appendix
\onecolumn
\section*{Supplementary material}

\section{Notations}
Let $(\mathcal{Z},d)$ be a Polish metric space (i.e. complete and separable). We say that $(\mathcal{Z},d)$ is proper if for all $z_0\in\mathcal{Z}$ and $R>0$, $B(z_0,R):=\{z\mid d(z,z_0)\leq R\}$ is compact. For $(\mathcal{Z},d)$ a Polish space, we denote $\mathcal{M}_+^1(\mathcal{Z})$ the set of Borel probability measures on $\mathcal{Z}$ endowed with $\lVert\cdot\rVert_{TV}$  strong topology. We recall the notion of weak topology: we say that a sequence $(\mu_n)_n$ of $\mathcal{M}_+^1(\mathcal{Z})$ converges weakly to $\mu\in\mathcal{M}_+^1(\mathcal{Z})$ if and only if for every continuous function $f$ on $\mathcal{Z}$, $\int fd\mu_n\to_{n\to\infty}\int f d\mu$. Endowed with its weak topology, $\mathcal{M}_+^1(\mathcal{Z})$ is a Polish space. For $\mu\in\mathcal{M}_+^1(\mathcal{Z})$, we define $L^1(\mu)$ the set of integrable functions with respect to $\mu$. We denote $\Pi_1:(z,z')\in\mathcal{Z}^2\mapsto z$ and $\Pi_2:(z,z')\in\mathcal{Z}^2\mapsto z'$ respectively the projections on the first and second component, which are continuous applications. For a measure $\mu$ and a measurable mapping $\mu$, we denote $g_\sharp\mu$ the pushforward measure of $\mu$ by $g$. Let $L\geq 1$ be an integer and denote $\Delta_L: = \{\lambda\in\mathbb{R}_+^L~\mathrm{s.t.}~\sum_{k=1}^L\lambda_k=1\}$, the probability simplex of $\mathbb{R}^L$. 
\section{Useful Lemmas}

\begin{lemma}[Fubini's theorem]
\label{lem:fubini}
Let $l:\Theta\times(\mathcal{X}\times\mathcal{Y})\rightarrow [0,\infty)$ satisfying Assumption~\ref{ass:loss}. Then for all $\mu\in\mathcal{M}^1_+(\Theta)$, $\int l(\theta,\cdot)d\mu(\theta)$ is Borel measurable; for  $\QQ\in\mathcal{M}^1_+(\mathcal{X}\times\mathcal{Y})$, $\int l(\cdot,(x,y))d\QQ(x,y)$ is Borel measurable. Moreover: $\int l(\theta,(x,y))d\mu(\theta)d\QQ(x,y)=\int l(\theta,(x,y))d\QQ(x,y)d\mu(\theta)$
\end{lemma}

\begin{lemma}
\label{lem:usc1}
Let $l:\Theta\times(\mathcal{X}\times\mathcal{Y})\rightarrow [0,\infty)$ satisfying Assumption~\ref{ass:loss}.
Then for all $\mu\in\mathcal{M}^1_+(\Theta)$, $(x,y)\mapsto\int l(\theta,(x,y))d\mu(\theta)$ is upper semi-continuous and hence Borel measurable.  
\end{lemma}
\begin{proof}
Let $(x_n,y_n)_n$ be a sequence of $\mathcal{X}\times\mathcal{Y}$ converging to $(x,y)\in\mathcal{X}\times\mathcal{Y}$.  For all $\theta\in\Theta$, $M-l(\theta,\cdot)$ is non negative and lower semi-continuous. Then by Fatou's Lemma applied:
\begin{align*}
   \int M-l(\theta,(x,y))d\mu(\theta)&\leq\int \liminf_{n\to\infty}  M-l(\theta,(x_n,y_n))d\mu(\theta)\\
   &\leq  \liminf_{n\to\infty}  \int M-l(\theta,(x_n,y_n))d\mu(\theta) 
\end{align*}

Then we deduce that: $\int M- l(\theta,\cdot)d\mu(\theta)$ is lower semi-continuous and then $\int l(\theta,\cdot)d\mu(\theta)$ is upper-semi continuous.
\end{proof}

\begin{lemma}
\label{lem:usc2}

Let $l:\Theta\times(\mathcal{X}\times\mathcal{Y})\rightarrow [0,\infty)$ satisfying Assumption~\ref{ass:loss}
Then for all $\mu\in\mathcal{M}^1_+(\Theta)$, $\QQ\mapsto\int l(\theta,(x,y))d\mu(\theta)d\QQ(x,y)$ is upper semi-continuous for weak topology of measures. 
\end{lemma}
\begin{proof}
 $-\int l(\theta,\cdot)d\mu(\theta) $ is lower semi-continuous from Lemma~\ref{lem:usc1}. Then $M-\int l(\theta,\cdot)d\mu(\theta) $ is lower semi-continuous and non negative. Let denote $v$ this function. Let $(v_n)_n$ be a non-decreasing sequence of continuous bounded functions such that $v_n\to v$. Let $(\QQ_k)_k$ converging weakly towards $\QQ$. Then by monotone convergence:
 
 \begin{align*}
     \int vd\QQ = \lim_n \int v_nd\QQ =\lim_n \lim_k\int v_nd\QQ_k\leq \liminf_k \int vd\QQ_k
 \end{align*}
 Then $\QQ\mapsto\int vd\QQ$ is lower semi-continuous and then $\QQ\mapsto\int l(\theta,(x,y))d\mu(\theta)d\QQ(x,y)$ is upper semi-continuous for weak topology of measures. 
 \end{proof}

\begin{lemma}
\label{lem:measure-sup}
Let $l:\Theta\times(\mathcal{X}\times\mathcal{Y})\rightarrow [0,\infty)$ satisfying Assumption~\ref{ass:loss}.
Then for all $\mu\in\mathcal{M}^1_+(\Theta)$, $(x,y)\mapsto \sup_{(x',y'),d(x,x')\leq\varepsilon,y=y'} \int l(\theta,(x',y'))d\mu(\theta)$ is universally measurable (i.e. measurable for all Borel probability measures). And hence the adversarial risk is well defined. 
\end{lemma}
\begin{proof}
Let $\phi :(x,y)\mapsto \sup_{(x',y'),d(x,x')\leq\varepsilon,y=y'} \int l(\theta,(x',y'))d\mu(\theta)$. Then for $u\in\bar{\mathbb{R}}$:
\begin{align*}
\left\{\phi(x,y)>u\right\}=\text{Proj}_1\left\{((x,y),(x',y'))\mid\int l(\theta,(x',y'))d\mu(\theta)-c_\varepsilon((x,y),(x',y'))>u\right\}
\end{align*}
By Lemma~\ref{lem:usc2}: $((x,y),(x',y'))\mapsto \int l(\theta,(x',y'))d\mu(\theta)-c_\varepsilon((x,y),(x',y'))$ is upper-semicontinuous hence Borel measurable. So its level sets are Borel sets, and by~\citep[Proposition 7.39]{bertsekas2004stochastic}, the projection of a Borel set is analytic. And then $\left\{\phi(x,y)>u\right\}$ universally measurable thanks to~\citep[Corollary 7.42.1]{bertsekas2004stochastic}. We deduce that $\phi$ is universally measurable.
\end{proof}

\section{Proofs}
\subsection{Proof of Proposition~\ref{prop:wass_ball}}
\label{prv:a_eps}
\begin{proof}
Let $\eta>0$. Let $\QQ\in\mathcal{A}_\varepsilon(\PP)$. There exists $\gamma\in
\mathcal{M}^+_1\left((\mathcal{X}\times\mathcal{Y})^2\right)$ such that, $d(x,x')\leq\varepsilon$, $y=y'$ $\gamma$-almost surely, and $\Pi_{1\sharp}\gamma=\PP$, and $\Pi_{2\sharp}\gamma=\QQ$. Then $\int c_\varepsilon d \gamma = 0\leq \eta$. Then, we deduce that $W_{c_\varepsilon}(\PP,\QQ)\leq \eta$, and $\QQ\in\mathcal{B}_{c_\varepsilon}(\PP,\eta)$. Reciprocally, let $\QQ\in\mathcal{B}_{c_\varepsilon}(\PP,\eta)$. Then, since the infimum is attained in the Wasserstein definition, there exists $\gamma\in
\mathcal{M}^+_1\left((\mathcal{X}\times\mathcal{Y})^2\right)$ such that $\int c_\varepsilon d \gamma \leq \eta$. Since $c_\varepsilon((x,x'),(y,y'))=+\infty$ when $d(x,x')>\varepsilon$ and $y\neq y'$, we deduce that, $d(x,x')\leq\varepsilon$ and $y=y'$, $\gamma$-almost surely. Then $\QQ\in\mathcal{A}_\varepsilon(\PP)$. We have then shown that: $\mathcal{A}_\varepsilon(\PP)=\mathcal{B}_{c_\varepsilon}(\PP,\eta)$.

The convexity of $\mathcal{A}_\varepsilon(\PP)$ is then immediate from the relation with the Wasserstein uncertainty set.

Let us show first that $\mathcal{A}_\varepsilon(\PP)$ is relatively compact for weak topology. To do so we will show that $\mathcal{A}_\varepsilon(\PP)$ is tight and apply Prokhorov's theorem. Let $\delta>0$, $(\mathcal{X}\times \mathcal{Y},d\oplus d')$ being a Polish space, $\{\PP\}$ is tight then there exists $K_\delta$ compact such that $\PP(K_\delta)\geq1-\delta$.
Let $\Tilde{K}_\delta:=\left\{(x',y')\mid \exists (x,y)\in K_\delta,~ d(x',x)\leq\varepsilon,~y=y'\right\}$.  Recalling that $(\mathcal{X},d)$ is proper (i.e. the closed balls are compact), so $\Tilde{K}_\delta$ is compact. Moreover for $\QQ\in\mathcal{A}_\varepsilon(\PP)$, $\QQ(\Tilde{K}_\delta)\geq \PP(K_\delta)\geq 1-\delta$. And then, Prokhorov's theorem holds, and $\mathcal{A}_\varepsilon(\PP)$ is relatively compact for weak topology.

Let us now prove that $\mathcal{A}_\varepsilon(\PP)$ is closed to conclude.  Let $(\QQ_n)_n$ be a sequence of $\mathcal{A}_\varepsilon(\PP)$ converging towards some $\QQ$ for weak topology. For each $n$, there exists $\gamma_n\in \mathcal{M}^1_+(\mathcal{X}\times\mathcal{Y})$ such that $d(x,x')\leq\varepsilon$ and $y=y'$ $\gamma_n$-almost surely and $\Pi_{1\sharp}\gamma_n=\PP$, $\Pi_{2\sharp}\gamma_n=\QQ_n$. $\{\QQ_n,n\geq0\}$ is relatively compact, then tight, then $\bigcup_n \Gamma_{\PP,\QQ_n}$ is tight, then relatively compact by Prokhorov's theorem. $(\gamma_n)_n\in\bigcup_n \Gamma_{\PP,\QQ_n}$, then up to an extraction,  $\gamma_n\to\gamma$. Then $d(x,x')\leq\varepsilon$ and $y=y'$ $\gamma$-almost surely, and by continuity, $\Pi_{1\sharp}\gamma=\PP$ and by continuity, $\Pi_{2\sharp}\gamma=\QQ$. And hence $\mathcal{A}_\varepsilon(\PP)$ is closed.

Finally $\mathcal{A}_\varepsilon(\PP)$ is a convex compact set for the weak topology. 
\end{proof}

\subsection{Proof of Proposition~\ref{prop:dro_adv}}
\label{prv:dro_adv}

\begin{proof}
Let $\mu\in\mathcal{M}^1_+(\Theta)$. Let $\Tilde{f}:((x,y),(x',y'))\mapsto \mathbb{E}_{\theta\sim\mu}\left[l(\theta,(x,y))\right]-c_\varepsilon((x,y),(x',y'))$. $\Tilde{f}$ is upper-semi continuous, hence upper semi-analytic. Then, by upper semi continuity of $\mathbb{E}_{\theta\sim\mu}\left[l(\theta,\cdot)\right]$ on the compact $\{(x',y')\mid~d(x,x')\leq\varepsilon,y=y'\}$ and~\citep[Proposition 7.50]{bertsekas2004stochastic}, there exists a universally measurable mapping $T$ such that $\mathbb{E}_{\theta\sim\mu}\left[l(\theta,T(x,y))\right]=\sup_{(x',y'),~d(x,x')\leq\varepsilon,y=y'}\mathbb{E}_{\theta\sim\mu}\left[l(\theta,(x,y))\right]$.  Let $\QQ = T_{\sharp}\PP$, then $\QQ\in\mathcal{A}_\varepsilon(\PP)$. And then $\mathbb{E}_{(x,y)\sim\PP}\left[\sup_{(x',y'),~d(x,x')\leq\varepsilon,y=y'}\mathbb{E}_{\theta\sim\mu}\left[l(\theta,(x',y'))\right]\right]\leq \sup_{\QQ\in\mathcal{A}_\varepsilon(\PP)}\mathbb{E}_{(x,y)\sim\QQ}\left[\mathbb{E}_{\theta\sim\mu}\left[l(\theta,(x,y))\right]\right]$.

Reciprocally, let $\QQ\in\mathcal{A}_\varepsilon(\PP)$. There exists $\gamma\in\mathcal{M}^1_+((\mathcal{X}\times\mathcal{Y})^2)$, such that $d(x,x')\leq\varepsilon$ and $y=y'$ $\gamma$-almost surely, and, $\Pi_{1\sharp}\gamma=\PP$ and  $\Pi_{2\sharp}\gamma=\QQ$. Then:
$\mathbb{E}_{\theta\sim\mu}\left[l(\theta,(x',y'))\right]\leq\sup_{(u,v),~d(x,u)\leq\varepsilon,y=v}\mathbb{E}_{\theta\sim\mu}\left[l(\theta,(u,v))\right]$ $\gamma$-almost surely. Then, we deduce that:
\begin{align*}
    \mathbb{E}_{(x',y')\sim\QQ}\left[\mathbb{E}_{\theta\sim\mu}\left[l(\theta,(x',y'))\right]\right]& =     \mathbb{E}_{(x,y,x',y')\sim\gamma}\left[\mathbb{E}_{\theta\sim\mu}\left[l(\theta,(x',y'))\right]\right] \\
    &\leq\mathbb{E}_{(x,y,x',y')\sim\gamma}\left[\sup_{(u,v),~d(x,u)\leq\varepsilon,y=v}\mathbb{E}_{\theta\sim\mu}\left[l(\theta,(u,v))\right]\right]\\
    &\leq\mathbb{E}_{(x,y)\sim\PP}\left[\sup_{(u,v),~d(x,u)\leq\varepsilon,y=v}\mathbb{E}_{\theta\sim\mu}\left[l(\theta,(u,v))\right]\right]
\end{align*}

Then we deduce the expected result:
\begin{align*}
\riskadv^\varepsilon(\mu)= \sup_{\QQ\in\mathcal{A}_\varepsilon(\PP)}\mathbb{E}_{(x,y)\sim\QQ}\left[\mathbb{E}_{\theta\sim\mu}\left[l(\theta,(x,y))\right]\right]
\end{align*}
Let us show that the optimum is attained. $\QQ\mapsto\mathbb{E}_{(x,y)\sim\QQ}\left[\mathbb{E}_{\theta\sim\mu}\left[l(\theta,(x,y))\right]\right]$ is upper semi continuous by Lemma~\ref{lem:usc2} for the weak topology of measures, and $\mathcal{A}_\varepsilon(\PP)$ is compact by Proposition~\ref{prop:wass_ball}, then by~\citep[Proposition 7.32]{bertsekas2004stochastic}, the supremum is attained for a certain $\QQ^*\in\mathcal{A}_\varepsilon(\PP)$.

\end{proof}

\subsection{Proof of Theorem~\ref{thm:duality-rand}}
\label{prv:duality-rand}
Let us first recall the Fan's Theorem.

\begin{thm}
Let $U$ be a compact convex Haussdorff space and $V$ be convex space (not necessarily topological). Let  $\psi:U\times V\to \mathbb{R}$ be a concave-convex function such that for all $v\in V$, $\psi(\cdot,v)$ is upper semi-continuous then:
\begin{align*}
    \inf_{v\in V}    \max_{u\in U}\psi(u,v) =    \max_{u\in U} \inf_{v\in V}    \psi(u,v) 
\end{align*}
\end{thm}

We are now set to prove Theorem~\ref{thm:duality-rand}.

\begin{proof}
$\mathcal{A}_\varepsilon(\PP)$, endowed with the weak topology of measures, is a Hausdorff compact convex space, thanks to Proposition~\ref{prop:wass_ball}. Moreover, $\mathcal{M}^1_+(\Theta)$ is clearly convex and $(\QQ,\mu)\mapsto \int ld\mu d\QQ$ is bilinear, hence concave-convex. Moreover thanks to Lemma~\ref{lem:usc2}, for all $\mu$, $\QQ\mapsto \int ld\mu d\QQ$ is upper semi-continuous. Then Fan's theorem applies and strong duality holds.
\end{proof}
In the related work (Section~\ref{sec:rw}), we mentioned a particular form of Theorem~\ref{thm:duality-rand} for convex cases. As mentioned, this result has limited impact in the adversarial classification setting. It is still a direct corollary of Fan's theorem. This theorem can be stated as follows: 
\begin{thm}Let $\PP\in\mathcal{M}^1_+(\mathcal{X}\times\mathcal{Y})$, $\varepsilon>0$ and $\Theta$ a convex set. Let $l$ be a loss satisfying Assumption~\ref{ass:loss}, and also, $(x,y)\in\mathcal{X}\times\mathcal{Y}$, $l(\cdot,(x,y))$ is a convex function, then we have the following:
\begin{align*}
\inf_{\theta\in\Theta} \sup_{\QQ\in \mathcal{A}_{\varepsilon}(\PP)} \mathbb{E}_{ \QQ }\left[l(\theta,(x,y))\right]
=
\sup_{\QQ\in \mathcal{A}_{\varepsilon}(\PP)}\inf_{\theta\in \Theta}  \mathbb{E}_{\QQ }\left[l(\theta,(x,y))\right]
\end{align*}
The supremum is always attained. If $\Theta$ is a compact set then, the infimum is also attained.
\end{thm}

\subsection{Proof of Proposition~\ref{prop:limit-eps}}
\label{prv:limit_eps}
\begin{proof}
Let us first show that for $\alpha\geq 0$, $\sup_{\QQ_i\in\Gamma_{i,\varepsilon}}\mathbb{E}_{\QQ_i,\mu}\left[l(\theta,(x,y))\right]-\alpha\text{KL}\left(\QQ_i\Big|\Big|\frac{1}{N}\mathbb{U}_{(x_i,y_i)}\right)$ admits a solution. Let $\alpha\geq 0$,
$(\QQ_{\alpha,i}^{n})_{n\geq 0}$ a sequence such that
\begin{align*}
  \mathbb{E}_{\QQ_{\alpha,i}^{n},\mu}\left[l(\theta,(x,y))\right]-\alpha\text{KL}\left(\QQ_{\alpha,i}^{n}\Big|\Big|\frac{1}{N}\mathbb{U}_{(x_i,y_i)}\right)\xrightarrow[n\to+\infty]{} \sup_{\QQ_i\in\Gamma_{i,\varepsilon}}\mathbb{E}_{\QQ_i,\mu}\left[l(\theta,(x,y))\right]-\alpha\text{KL}\left(\QQ_i\Big|\Big|\frac{1}{N}\mathbb{U}_{(x_i,y_i)}\right).
\end{align*}
As $\Gamma_{i,\varepsilon}$ is tight ($(\mathcal{X},d)$ is a proper metric space therefore all the closed ball are compact) and by Prokhorov's theorem, we can extract a subsequence which converges  toward  $\QQ^{*}_{\alpha,i}$. Moreover, $l$ is upper semi-continuous (u.s.c), thus $\QQ\rightarrow \mathbb{E}_{\QQ,\mu}\left[l(\theta,(x,y))\right]$ is also u.s.c.\footnote{Indeed by considering a decreasing sequence of continuous and bounded functions which converge towards $\mathbb{E}_{\mu}\left[l(\theta,(x,y))\right]$ and by definition of the weak convergence the result follows.} Moreover 
$\QQ\rightarrow - \alpha \text{KL}\left(\QQ\Big|\Big|\frac{1}{N}\mathbb{U}_{(x_i,y_i)}\right)$ is also u.s.c. \footnote{for $\alpha=0$ the result is clear, and if $\alpha>0$, note that $\text{KL}\left(\cdot\Big|\Big|\frac{1}{N}\mathbb{U}_{(x_i,y_i)}\right)$ is lower semi-continuous}, therefore, by considering the limit superior as $n$ goes to infinity we obtain that
\begin{align*}
    &\limsup_{n\to+\infty}\mathbb{E}_{\QQ_{\alpha,i}^{n},\mu}\left[l(\theta,(x,y))\right]-\alpha\text{KL}\left(\QQ_{\alpha,i}^{n}\Big|\Big|\frac{1}{N}\mathbb{U}_{(x_i,y_i)}\right)\\
    &=\sup_{\QQ_i\in\Gamma_{i,\varepsilon}}\mathbb{E}_{\QQ_i,\mu}\left[l(\theta,(x,y))\right]-\alpha\text{KL}\left(\QQ_i\Big|\Big|\frac{1}{N}\mathbb{U}_{(x_i,y_i)}\right)\\
    &\leq \mathbb{E}_{\QQ_{\alpha,i}^{*},\mu}\left[l(\theta,(x,y))\right]-\alpha\text{KL}\left(\QQ_{\alpha,i}^{*}\Big|\Big|\frac{1}{N}\mathbb{U}_{(x_i,y_i)}\right)
\end{align*}
from which we deduce that $\QQ_{\alpha,i}^{*}$ is optimal.

Let us now show the result. We consider a positive sequence of $(\alpha_i^{(\ell)})_{\ell\geq0}$ such that $\alpha_i^{(\ell)}\to 0$.
Let us denote $\QQ^{*}_{\alpha_i^{(\ell)},i}$ and $\QQ^{*}_i$ the solutions of  $\max_{\QQ_i\in\Gamma_{i,\varepsilon}}\mathbb{E}_{\QQ_i,\mu}\left[l(\theta,(x,y))\right]-\alpha_i^{(\ell)}\text{KL}\left(\QQ_i\Big|\Big|\frac{1}{N}\mathbb{U}_{(x_i,y_i)}\right)$
and 
$\max_{\QQ_i\in\Gamma_{i,\varepsilon}}\mathbb{E}_{\QQ_i,\mu}\left[l(\theta,(x,y))\right]$ respectively.  Since $\Gamma_{i,\varepsilon}$ is tight, $(\QQ^{*}_{\alpha_i^{(\ell)},i})_{\ell\geq 0}$ is also tight and we can extract by Prokhorov's theorem a subsequence which converges towards $\QQ^{*}$. Moreover we have
\begin{align*}
 \mathbb{E}_{\QQ^{*}_i,\mu}\left[l(\theta,(x,y))\right] -\alpha_i^{(\ell)}\text{KL}\left(\QQ^{*}_i\Big|\Big|\frac{1}{N}\mathbb{U}_{(x_i,y_i)}\right)\leq \mathbb{E}_{\QQ^{*}_{\alpha_i^{(\ell)},i},\mu}\left[l(\theta,(x,y))\right] -\alpha_i^{(\ell)}\text{KL}\left(\QQ^{*}_{\alpha_i^{(\ell)},i}\Big|\Big|\frac{1}{N}\mathbb{U}_{(x_i,y_i)}\right)
\end{align*}
from which follows that
\begin{align*}
0\leq \mathbb{E}_{\QQ^{*}_i,\mu}\left[l(\theta,(x,y))\right] -  \mathbb{E}_{\QQ^{*}_{\alpha_i^{(\ell)},i},\mu}\left[l(\theta,(x,y))\right]\leq \alpha_i^{(\ell)}\left(\text{KL}\left(\QQ^{*}_i\Big|\Big|\frac{1}{N}\mathbb{U}_{(x_i,y_i)}\right)- \text{KL}\left(\QQ^{*}_{\alpha_i^{(\ell)},i}\Big|\Big|\frac{1}{N}\mathbb{U}_{(x_i,y_i)}\right)\right)
\end{align*}
Then by considering the limit superior we obtain that
\begin{align*}
    \limsup_{\ell\to+\infty}\mathbb{E}_{\QQ^{*}_{\alpha_i^{(\ell)},i},\mu}\left[l(\theta,(x,y))\right] = \mathbb{E}_{\QQ^{*}_i,\mu}\left[l(\theta,(x,y))\right].
\end{align*}
from which follows that 
\begin{align*}
 \mathbb{E}_{\QQ^{*}_i,\mu}\left[l(\theta,(x,y))\right]\leq \mathbb{E}_{\QQ^{*},\mu}\left[l(\theta,(x,y))\right]
\end{align*}
and by optimality of $\QQ^{*}_i$ we obtain the desired result. 
\end{proof}

\subsection{Proof of Proposition~\ref{prop:control-error-stat}}
\label{prv:control-error-stat}

\begin{proof}
Let us denote for all $\mu\in\mathcal{M}_1^{+}(\Theta)$,
\begin{align*}
  \widehat{\mathcal{R}}^{\bm{\varepsilon},\textbf{m}}_{adv,\bm{\alpha}}(\mu):=  \sum_{i=1}^N  \frac{\alpha_i}{N}\log\left( \frac{1}{m_i}\sum_{j=1}^{m_i}\exp\frac{\mathbb{E}_{ \mu}\left[l(\theta,u_j^{(i)})\right]}{\alpha_i}\right).
\end{align*}
Let also consider $(\mu^{(\textbf{m})}_n)_{n\geq 0}$ and $(\mu_n)_{n\geq 0}$ two sequences such that
\begin{align*}
 \widehat{\mathcal{R}}^{\bm{\varepsilon},\textbf{m}}_{adv,\bm{\alpha}}(\mu^{(\textbf{m})}_n) \xrightarrow[n \to +\infty]{}\widehat{\mathcal{R}}^{\bm{\varepsilon},\textbf{m}}_{adv,\bm{\alpha}},~\quad
\widehat{\mathcal{R}}^{\bm{\varepsilon}}_{adv,\bm{\alpha}}(\mu_n)\xrightarrow[n \to +\infty]{}\widehat{\mathcal{R}}^{\bm{\varepsilon},*}_{adv,\bm{\alpha}}.
\end{align*}
We first remarks that
\begin{align*}
\widehat{\mathcal{R}}^{\bm{\varepsilon},\textbf{m}}_{adv,\bm{\alpha}}- \widehat{\mathcal{R}}^{\bm{\varepsilon},*}_{adv,\bm{\alpha}}&\leq \widehat{\mathcal{R}}^{\bm{\varepsilon},\textbf{m}}_{adv,\bm{\alpha}} - \widehat{\mathcal{R}}^{\bm{\varepsilon},\textbf{m}}_{adv,\bm{\alpha}}(\mu_n) + \widehat{\mathcal{R}}^{\bm{\varepsilon},\textbf{m}}_{adv,\bm{\alpha}}(\mu_n) - \widehat{\mathcal{R}}^{\bm{\varepsilon}}_{adv,\bm{\alpha}}(\mu_n)+ \widehat{\mathcal{R}}^{\bm{\varepsilon}}_{adv,\bm{\alpha}}(\mu_n)-
\widehat{\mathcal{R}}^{\bm{\varepsilon},*}_{adv,\bm{\alpha}} \\
&\leq \sup_{\mu\in \mathcal{M}^+_1(\Theta)}\Big|\widehat{\mathcal{R}}^{\bm{\varepsilon},\textbf{m}}_{adv,\bm{\alpha}}(\mu) - \widehat{\mathcal{R}}^{\bm{\varepsilon}}_{adv,\bm{\alpha}}(\mu) \Big| + \widehat{\mathcal{R}}^{\bm{\varepsilon}}_{adv,\bm{\alpha}}(\mu_n)-
\widehat{\mathcal{R}}^{\bm{\varepsilon},*}_{adv,\bm{\alpha}},
\end{align*}
and by considering the limit, we obtain that
\begin{align*}
  \widehat{\mathcal{R}}^{\bm{\varepsilon},\textbf{m}}_{adv,\bm{\alpha}}- \widehat{\mathcal{R}}^{\bm{\varepsilon},*}_{adv,\bm{\alpha}}&\leq  \sup_{\mu\in \mathcal{M}^+_1(\Theta)}\Big|\widehat{\mathcal{R}}^{\bm{\varepsilon},\textbf{m}}_{adv,\bm{\alpha}}(\mu) - \widehat{\mathcal{R}}^{\bm{\varepsilon}}_{adv,\bm{\alpha}}(\mu) \Big| 
\end{align*}
Simarly we have that
\begin{align*}
\widehat{\mathcal{R}}^{\bm{\varepsilon},*}_{adv,\bm{\alpha}} - \widehat{\mathcal{R}}^{\bm{\varepsilon},\textbf{m}}_{adv,\bm{\alpha}}&\leq \widehat{\mathcal{R}}^{\bm{\varepsilon},*}_{adv,\bm{\alpha}} -
\widehat{\mathcal{R}}^{\bm{\varepsilon}}_{adv,\bm{\alpha}}(\mu_n^{(\bm{m})})
+\widehat{\mathcal{R}}^{\bm{\varepsilon}}_{adv,\bm{\alpha}}(\mu_n^{(\bm{m})}) - \widehat{\mathcal{R}}^{\bm{\varepsilon},\textbf{m}}_{adv,\bm{\alpha}}(\mu_n^{(\bm{m})}) + \widehat{\mathcal{R}}^{\bm{\varepsilon},\textbf{m}}_{adv,\bm{\alpha}}(\mu_n^{(\bm{m})}) - \widehat{\mathcal{R}}^{\bm{\varepsilon},\textbf{m}}_{adv,\bm{\alpha}}
\end{align*}
from which follows that 
\begin{align*}
\widehat{\mathcal{R}}^{\bm{\varepsilon},*}_{adv,\bm{\alpha}} - \widehat{\mathcal{R}}^{\bm{\varepsilon},\textbf{m}}_{adv,\bm{\alpha}}&\leq  \sup_{\mu\in \mathcal{M}^+_1(\Theta)}\Big|\widehat{\mathcal{R}}^{\bm{\varepsilon},\textbf{m}}_{adv,\bm{\alpha}}(\mu) - \widehat{\mathcal{R}}^{\bm{\varepsilon}}_{adv,\bm{\alpha}}(\mu) \Big| 
\end{align*}
Therefore we obtain that 
\begin{align*}
\Big| \widehat{\mathcal{R}}^{\bm{\varepsilon},*}_{adv,\bm{\alpha}} - \widehat{\mathcal{R}}^{\bm{\varepsilon},\textbf{m}}_{adv,\bm{\alpha}}\Big |\leq 
\sum_{i=1}^N\frac{\alpha}{N}  &\sup_{\mu\in \mathcal{M}^+_1(\Theta)}\Big|\log\left(\frac{1}{m_i}\sum_{j=1}^{m_i}\exp\left(\frac{\mathbb{E}_{\theta \sim \mu}\left[l(\theta,u_j^{(i)}))\right]}{\alpha}\right)\right)\\
    &- \log\left(\int_{\mathcal{X}\times\mathcal{Y}}\exp\left(\frac{\mathbb{E}_{\theta \sim \mu}\left[l(\theta,(x,y))\right]}{\alpha}\right) d\mathbb{U}_{(x_i,y_i)} \right)\Big|.
\end{align*}
Observe that $l\geq 0$, therefore because the $\log$ function is 1-Lipschitz on $[1,+\infty)$, we obtain that 
\begin{align*}
\Big| \widehat{\mathcal{R}}^{\bm{\varepsilon},*}_{adv,\bm{\alpha}} - \widehat{\mathcal{R}}^{\bm{\varepsilon},\textbf{m}}_{adv,\bm{\alpha}}\Big |\leq 
\sum_{i=1}^N\frac{\alpha}{N}  \sup_{\mu\in \mathcal{M}^+_1(\Theta)}\Big | \frac{1}{m_i}\sum_{j=1}^{m_i}\exp\left(\frac{\mathbb{E}_{\theta \sim \mu}\left[l(\theta,u_j^{(i)}))\right]}{\alpha}\right)
    - \int_{\mathcal{X}\times\mathcal{Y}}\exp\left(\frac{\mathbb{E}_{\theta \sim \mu}\left[l(\theta,(x,y))\right]}{\alpha}\right) d\mathbb{U}_{(x_i,y_i)} \Big|.
\end{align*}
Let us now denote for all $i=1,\dots,N$,
\begin{align*}
    \widehat{R}_i(\mu,\bm{u}^{(i)})&:=\sum_{j=1}^{m_i}\exp\left(\frac{\mathbb{E}_{\theta \sim \mu}\left[l(\theta,u_j^{(i)}))\right]}{\alpha}\right)\\
    R_i(\mu)&:= \int_{\mathcal{X}\times\mathcal{Y}}\exp\left(\frac{\mathbb{E}_{\theta \sim \mu}\left[l(\theta,(x,y))\right]}{\alpha}\right) d\mathbb{U}_{(x_i,y_i)}.
\end{align*}
and let us define 
\begin{align*}
    f(\bm{u}^{(1)},\dots,\bm{u}^{(N)}):=\sum_{i=1}^N\frac{\alpha}{N}\sup_{\mu\in \mathcal{M}^+_1(\Theta)}\Big |\widehat{R}_i(\mu) -R_i(\mu)\Big |
\end{align*}
where $\bm{u}^{(i)}:=(u_1^{(i)},\dots,u_1^{(m)})$. By denoting $z^{(i)}=(u_1^{(i)},\dots,u_{k-1}^{(i)},z,u_{k+1}^{(i)},\dots,u_m^{(i)})$, we have that
\begin{align*}
  |f(\bm{u}^{(1)},\dots,\bm{u}^{(N)}) - f(\bm{u}^{(1)},\dots,\bm{u}^{(i-1)},\bm{z}^{(i)},\bm{u}^{(i+1)},\dots,\bm{u}^{(N)})|&\leq \frac{\alpha}{N}\Big | \sup_{\mu\in \mathcal{M}^+_1(\Theta)}\Big |\widehat{R}_i(\mu,\bm{u}^{(i)}) -R_i(\mu)\Big |\\
  &- \sup_{\mu\in \mathcal{M}^+_1(\Theta)}\Big |\widehat{R}_i(\mu,\bm{z}^{(i)}) -R_i(\mu)\Big | \Big |\\
  &\leq \frac{\alpha}{N}\Big |\frac{1}{m}\left[\exp\left(\frac{\mathbb{E}_{\theta \sim \mu}\left[l(\theta,u_k^{(i)}))\right]}{\alpha}\right) - \exp\left(\frac{\mathbb{E}_{\theta \sim \mu}\left[l(\theta,z^{(i)}))\right]}{\alpha}\right) \right]\Big| \\
  &\leq \frac{2\exp(M/\alpha)}{Nm}
\end{align*}
where the last inequality comes from the fact that the loss is upper bounded by $l\leq M$. Then by appling the McDiarmid’s Inequality, we obtain that with a probability of at least $1-\delta$,
\begin{align*}
 \Big| \widehat{\mathcal{R}}^{\bm{\varepsilon},*}_{adv,\bm{\alpha}} - \widehat{\mathcal{R}}^{\bm{\varepsilon},\textbf{m}}_{adv,\bm{\alpha}}\Big |\leq\mathbb{E}(f(\bm{u}^{(1)},\dots,\bm{u}^{(N)}))+\frac{2\exp(M/\alpha)}{\sqrt{mN}}\sqrt{\frac{\log(2/\delta)}{2}}.
\end{align*}
Thanks to~\citep[Lemma 26.2]{shalev2014understanding}, we have for all $i\in\{1,\dots,N\}$
\begin{align*}
\mathbb{E}(f(\bm{u}^{(1)},\dots,\bm{u}^{(N)}))\leq 2 \mathbb{E}(\text{Rad}(\mathcal{F}_i\circ \mathbf{u^{(i)}}))    
\end{align*}
where for any class of function $\mathcal{F}$ defined on $\mathcal{Z}$  and point $\bm{z}:(z_1,\dots,z_q)\in\mathcal{Z}^q$
\begin{align*}
    &\mathcal{F}\circ \bm{z}:=\Big\{(f(z_1),\dots,f(z_q)),~f\in\mathcal{F}\Big\} \quad,\quad \text{Rad}(\mathcal{F}\circ \bm{z}):=\frac{1}{q}\mathbb{E}_{\bm{\sigma}\sim\{\pm 1\}}\left[\sup_{f\in\mathcal{F}}\sum_{i=1}^q\sigma_if(z_i)\right]\\
    &\mathcal{F}_i:=\Big\{u\rightarrow\exp\left(\frac{\mathbb{E}_{\theta \sim \mu}\left[l(\theta,u))\right]}{\alpha}\right),~\mu\in\mathcal{M}_{1}^{+}(\Theta) \Big\}.
      \end{align*}
Moreover as $x\rightarrow\exp(x/\alpha)$ is $\frac{\exp(M/\alpha)}{\alpha}$-Lipstchitz on $(-\infty,M]$, by~\citep[Lemma 26.9]{shalev2014understanding}, we have 
\begin{align*}
   \text{Rad}(\mathcal{F}_i\circ \mathbf{u^{(i)}})\leq \frac{\exp(M/\alpha)}{\alpha} \text{Rad}(\mathcal{H}_i\circ \mathbf{u^{(i)}}) 
\end{align*}
where 
\begin{align*}
    \mathcal{H}_i:=\Big\{u\rightarrow \mathbb{E}_{\theta \sim \mu}\left[l(\theta,u))\right],~\mu\in\mathcal{M}_{1}^{+}(\Theta) \Big\}.
\end{align*}
Let us now define
\begin{align*}
    g(\bm{u}^{(1)},\dots,\bm{u}^{(N)}):=\sum_{j=1}^N\frac{2\exp(M/\alpha)}{N}\text{Rad}(\mathcal{H}_j\circ \mathbf{u^{(j)}}).
\end{align*}
We observe that 
\begin{align*}
|g(\bm{u}^{(1)},\dots,\bm{u}^{(N)}) - g(\bm{u}^{(1)},\dots,\bm{u}^{(i-1)},\bm{z}^{(i)},\bm{u}^{(i+1)},\dots,\bm{u}^{(N)})|&\leq \frac{2\exp(M/\alpha)}{N}|\text{Rad}(\mathcal{H}_i\circ \mathbf{u^{(i)}}) - \text{Rad}(\mathcal{H}_i\circ \mathbf{z^{(i)}})|\\
&\leq \frac{2\exp(M/\alpha)}{N}\frac{2M}{m}.
\end{align*}
By Applying the McDiarmid’s Inequality, we have that with a probability of at least $1-\delta$
\begin{align*}
\mathbb{E}(g(\bm{u}^{(1)},\dots,\bm{u}^{(N)}))\leq g(\bm{u}^{(1)},\dots,\bm{u}^{(N)}) +\frac{4\exp(M/\alpha)M}{\sqrt{mN}}\sqrt{\frac{\log(2/\delta)}{2}}.
\end{align*}
Remarks also that 
\begin{align*}
    \text{Rad}(\mathcal{H}_i\circ \mathbf{u^{(i)}})&=\frac{1}{m}\mathbb{E}_{\bm{\sigma}\sim\{\pm 1\}}\left[\sup_{\mu\in\mathcal{M}_1^{+}(\Theta)}\sum_{j=1}^m\sigma_i\mathbb{E}_{\mu}(l(\theta,u^{(i)}_j))\right]\\
    &=\frac{1}{m}\mathbb{E}_{\bm{\sigma}\sim\{\pm 1\}}\left[\sup_{\theta\in\Theta}\sum_{j=1}^m\sigma_i l(\theta,u^{(i)}_j)\right]
\end{align*}
Finally, applying a union bound leads to the desired result.

\end{proof}

\subsection{Proof of Proposition~\ref{prop:control-error-approx}}
\label{prv:control-error-approx}
\begin{proof}
Following the same steps than the proof of Proposition~\ref{prop:control-error-stat}, let $(\mu_n^{\varepsilon})_{n\geq 0}$ and $(\mu_n)_{n\geq 0}$ two sequences such that
\begin{align*}
    \widehat{\mathcal{R}}_{adv,\bm{\alpha}}^{\varepsilon}(\mu_n^{\varepsilon})\xrightarrow[n \to +\infty]{}\widehat{\mathcal{R}}_{adv,\bm{\alpha}}^{\varepsilon,*},~\quad \widehat{\mathcal{R}}_{adv}^{\varepsilon}(\mu_n)\xrightarrow[n \to +\infty]{}\widehat{\mathcal{R}}_{adv}^{\varepsilon,*}.
\end{align*}
Remarks that 
\begin{align*}
  \widehat{\mathcal{R}}_{adv,\bm{\alpha}}^{\varepsilon,*} - \widehat{\mathcal{R}}_{adv}^{\varepsilon,*}&\leq \widehat{\mathcal{R}}_{adv,\bm{\alpha}}^{\varepsilon,*} - \widehat{\mathcal{R}}_{adv,\bm{\alpha}}^{\varepsilon}(\mu_n) + \widehat{\mathcal{R}}_{adv,\bm{\alpha}}^{\varepsilon}(\mu_n) -   \widehat{\mathcal{R}}_{adv}^{\varepsilon}(\mu_n)+ \widehat{\mathcal{R}}_{adv}^{\varepsilon}(\mu_n)-\widehat{\mathcal{R}}_{adv}^{\varepsilon,*}\\
  &\leq \sup_{\mu\in\mathcal{M}_1^{+}(\Theta)}\Big|\widehat{\mathcal{R}}_{adv,\bm{\alpha}}^{\varepsilon}(\mu) -   \widehat{\mathcal{R}}_{adv}^{\varepsilon}(\mu)  \Big| + \widehat{\mathcal{R}}_{adv}^{\varepsilon}(\mu_n)-\widehat{\mathcal{R}}_{adv}^{\varepsilon,*}
\end{align*}
Then by considering the limit we obtain that 
\begin{align*}
    \widehat{\mathcal{R}}_{adv,\bm{\alpha}}^{\varepsilon,*} - \widehat{\mathcal{R}}_{adv}^{\varepsilon,*}&\leq \sup_{\mu\in\mathcal{M}_1^{+}(\Theta)}\Big|\widehat{\mathcal{R}}_{adv,\bm{\alpha}}^{\varepsilon}(\mu) -   \widehat{\mathcal{R}}_{adv}^{\varepsilon}(\mu)  \Big|.
\end{align*}
Similarly, we obtain that 
\begin{align*}
     \widehat{\mathcal{R}}_{adv}^{\varepsilon,*}-\widehat{\mathcal{R}}_{adv,\bm{\alpha}}^{\varepsilon,*}&\leq \sup_{\mu\in\mathcal{M}_1^{+}(\Theta)}\Big|\widehat{\mathcal{R}}_{adv,\bm{\alpha}}^{\varepsilon}(\mu) -   \widehat{\mathcal{R}}_{adv}^{\varepsilon}(\mu)  \Big|,
\end{align*}
from which follows that
\begin{align*}
 \Big| \widehat{\mathcal{R}}_{adv,\bm{\alpha}}^{\varepsilon,*} - \widehat{\mathcal{R}}_{adv}^{\varepsilon,*}\Big|&\leq \frac{1}{N}\sum_{i=1}^N\sup_{\mu\in\mathcal{M}_1^{+}(\Theta)}\Big|\alpha\log\left(\int_{\mathcal{X}\times\mathcal{Y}}\exp\left(\frac{\mathbb{E}_{\mu}[l(\theta,(x,y))]}{\alpha} \right) d\mathbb{U}_{(x_i,y_i)}\right)-\sup_{u\in S^{\varepsilon}_{(x_i,y_i)}}\mathbb{E}_{\mu}[l(\theta,u)] \Big|.
\end{align*}
Let $\mu\in\mathcal{M}_1^{+}(\Theta)$ and $i\in\{1,\dots,N\}$, then we have
\begin{align*}
 \Big|\alpha&\log\left(\int_{\mathcal{X}\times\mathcal{Y}}\exp\left(\frac{\mathbb{E}_{\mu}[l(\theta,(x,y))]}{\alpha} \right) d\mathbb{U}_{(x_i,y_i)}\right)-\sup_{u\in S^{\varepsilon}_{(x_i,y_i)}}\mathbb{E}_{\mu}[l(\theta,u)] \Big|\\
 &=\Big|\alpha\log\left(\int_{\mathcal{X}\times\mathcal{Y}}\exp\left(\frac{\mathbb{E}_{\mu}[l(\theta,(x,y))]-\sup_{u\in S^{\varepsilon}_{(x_i,y_i)}}\mathbb{E}_{\mu}[l(\theta,u)]}{\alpha} \right) d\mathbb{U}_{(x_i,y_i)}\right) \Big|  \\
 &=\alpha  \Big| \log\left(\int_{A_{\beta,\mu}^{(x_i,y_i)}}\exp\left(\frac{\mathbb{E}_{\mu}[l(\theta,(x,y))]-\sup_{u\in S^{\varepsilon}_{(x_i,y_i)}}\mathbb{E}_{\mu}[l(\theta,u)]}{\alpha} \right) d\mathbb{U}_{(x_i,y_i)}\right. \\
 &+ \left.\int_{(A_{\beta,\mu}^{(x_i,y_i)})^{c}}\exp\left(\frac{\mathbb{E}_{\mu}[l(\theta,(x,y))]-\sup_{u\in S^{\varepsilon}_{(x_i,y_i)}}\mathbb{E}_{\mu}[l(\theta,u)]}{\alpha} \right) d\mathbb{U}_{(x_i,y_i)}\right)  \Big|\\
 &\leq \alpha \Big | \log\left(\exp(-\beta/\alpha)\mathbb{U}_{(x_i,y_i)}\left(A_{\beta,\mu}^{(x_i,y_i)}\right) \right)\Big | \\
 &+ \alpha  \Big|\log\left(1+ \frac{\exp(\beta/\alpha)}{\mathbb{U}_{(x_i,y_i)}\left(A_{\beta,\mu}^{(x_i,y_i)}\right)}\int_{(A_{\beta,\mu}^{(x_i,y_i)})^{c}}\exp\left(\frac{\mathbb{E}_{\mu}[l(\theta,(x,y))]-\sup_{u\in S^{\varepsilon}_{(x_i,y_i)}}\mathbb{E}_{\mu}[l(\theta,u)]}{\alpha} \right) d\mathbb{U}_{(x_i,y_i)}\right)  \Big|\\
 &\leq \alpha\log(1/C_\beta)+\beta +\frac{\alpha}{C_\beta}\\
 &\leq 2\alpha\log(1/C_\beta)+\beta
\end{align*}
\end{proof}

\subsection{Proof of Proposition~\ref{prop:algo-oracle}}
\label{prv:algo-oracle}


\begin{proof}
Thanks to Danskin theorem, if $\QQ^*$ is a best response to $\bm{\lambda}$, then $\bm{g}^*:=\left(\mathbb{E}_{\QQ^*}\left[l(\theta_1,(x,y))\right],\dots,\mathbb{E}_{\QQ^*}\left[l(\theta_L,(x,y))\right]\right)^T$ is a subgradient of $\bm{\lambda}\to \riskadv^\varepsilon(\bm{\lambda})$. Let $\eta\geq 0$ be the learning rate. Then we have for all $t\geq 1$:
\begin{align*}
\lVert \bm{\lambda}_t-\bm{\lambda}^*\rVert^2&\leq \lVert \bm{\lambda}_{t-1}-\eta \bm{g}_t-\bm{\lambda}^*\rVert^2\\
&=\lVert \bm{\lambda}_{t-1}-\bm{\lambda}^*\rVert^2-2\eta \langle\bm{g}_t, \bm{\lambda}_{t-1}-\bm{\lambda}^*\rangle+ \eta^2\lVert \bm{g}_t\rVert^2\\
&\leq \lVert \bm{\lambda}_{t-1}-\bm{\lambda}^*\rVert^2-2\eta \langle\bm{g}^*_t, \bm{\lambda}_{t-1}-\bm{\lambda}^*\rangle+2\eta\langle\bm{g}^*_t-\bm{g}_t, \bm{\lambda}_{t-1}-\bm{\lambda}^*\rangle+\eta^2 M^2 L\\
&\leq \lVert \bm{\lambda}_{t-1}-\bm{\lambda}^*\rVert^2-2\eta\left(\riskadv^\varepsilon(\bm{\lambda}_t)-\riskadv^\varepsilon(\bm{\lambda}^*)\right) +4\eta\delta+\eta^2  M^2 L
\end{align*}
We then deduce by summing:
\begin{align*}
   2\eta \sum_{t=1}^T \riskadv^\varepsilon(\bm{\lambda}_t)-\riskadv^\varepsilon(\bm{\lambda}^*) \leq 4\delta\eta T +\lVert \bm{\lambda}_{0}-\bm{\lambda}^*\rVert^2+\eta^2 M^2 LT
\end{align*}
Then we have:
\begin{align*}
    \min_{t\in[T]}\riskadv^\varepsilon(\bm{\lambda}_t)-\riskadv^\varepsilon(\bm{\lambda}^*)\leq 2\delta+\frac{4}{\eta T}+M^2L\eta
\end{align*}
The left-hand term is minimal for $\eta=\frac{2}{M\sqrt{LT}}$, and for this value:
\begin{align*}
    \min_{t\in[T]}\riskadv^\varepsilon(\bm{\lambda}_t)-\riskadv^\varepsilon(\bm{\lambda}^*)\leq 2\delta+\frac{2M\sqrt{L}}{\sqrt{T}}
\end{align*}
\end{proof}.

\section{Additional Experimental Results}
\label{sec:additional-xp}

\subsection{Experimental setting.}

\paragraph{Optimizer.} For each of our models, The optimizer we used in all our implementations is SGD with learning rate set to $0.4$ at epoch $0$ and is divided by $10$ at half training then by $10$ at the three quarters of training. The momentum is set to $0.9$ and the weight decay to $5\times10^{-4}$. The batch size is set to $1024$. 
\paragraph{Adaptation of Attacks.} Since our classifier is randomized, we need to adapt the attack accordingly. To do so we used the expected loss:
\begin{align*}
\Tilde{l}\left((\bm{\lambda},\bm{\theta}),(x,y)\right) = \sum_{k=1}^L \lambda_k l(\theta_k,(x,y))
\end{align*}
to compute the gradient in the attacks, regardless the loss (DLR or cross-entropy). For the inner maximization at training time, we used a PGD attack on the cross-entropy loss with $\varepsilon=0.03$. For the final evaluation, we used the untargeted $DLR$ attack with default parameters.
\paragraph{Regularization in Practice.} The entropic regularization in higher dimensional setting need to be adapted to be more likely to find adversaries. To do so, we computed PGD attacks with only $3$ iterations with $5$ different restarts instead of sampling uniformly $5$ points  in the $\ell_\infty$-ball. In our experiments in the main paper, we use a regularization parameter $\alpha=0.001$. The learning rate for the minimization on $\bm{\lambda}$ is always fixed to $0.001$. 
\paragraph{Alternate Minimization Parameters.} Algorithm~\ref{algo:heuristic} implies an alternate minimization algorithm. We set the number of updates of $\bm{\theta}$ to $T_\theta = 50$ and, the update of $\bm{\lambda}$ to $T_\lambda = 25$. 

\subsection{Effect of the Regularization}
In this subsection, we experimentally investigate the effect of the regularization. In Figure~\ref{fig:xp-regularization}, we notice, that the regularization has the effect of stabilizing, reducing the variance and improving the level of the robust accuracy for adversarial training for mixtures (Algorithm~\ref{algo:heuristic}). The standard accuracy curves are very similar in both cases.

\begin{figure}[h]
    \centering
    \includegraphics[width=0.24\textwidth]{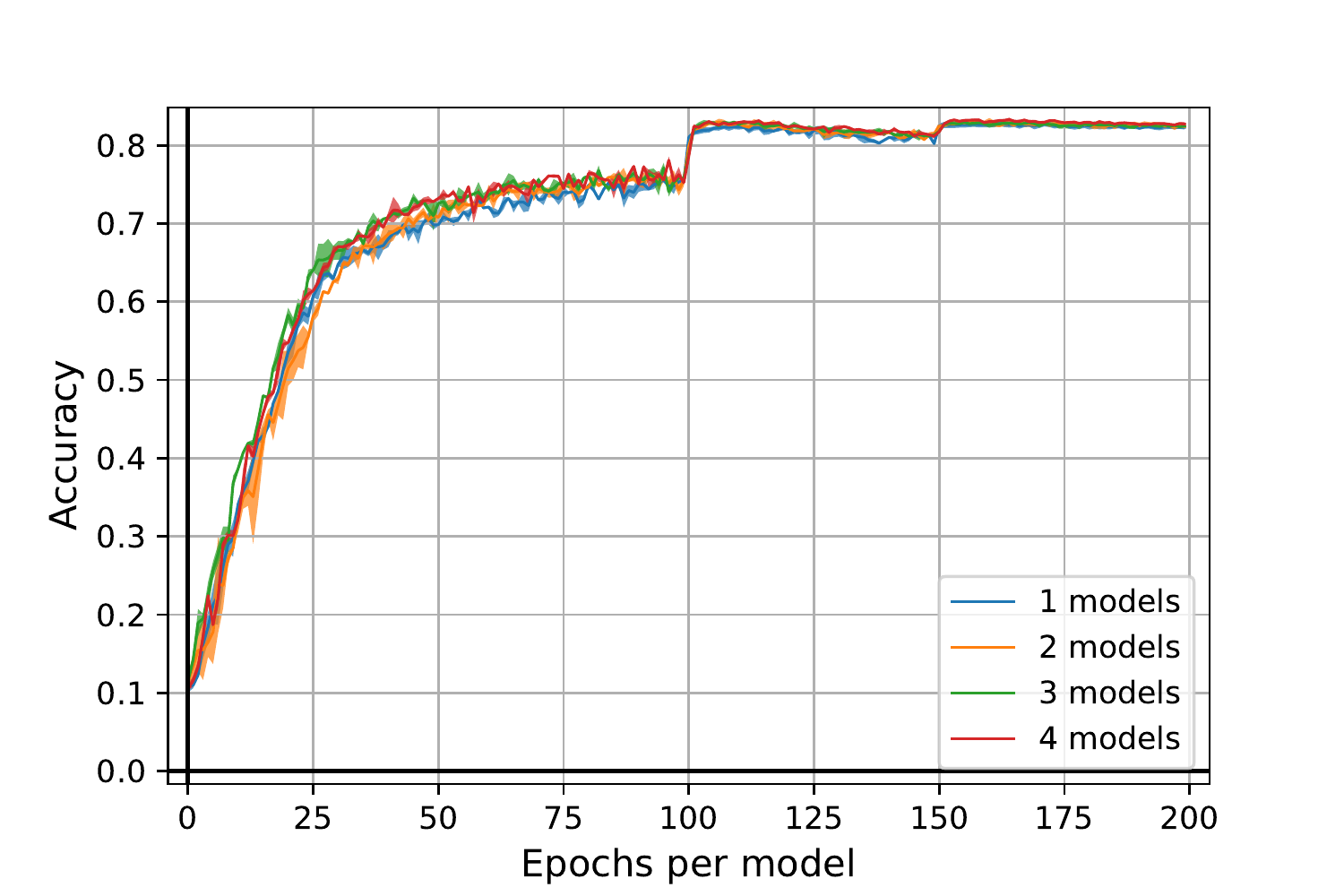}    \includegraphics[width=0.24\textwidth]{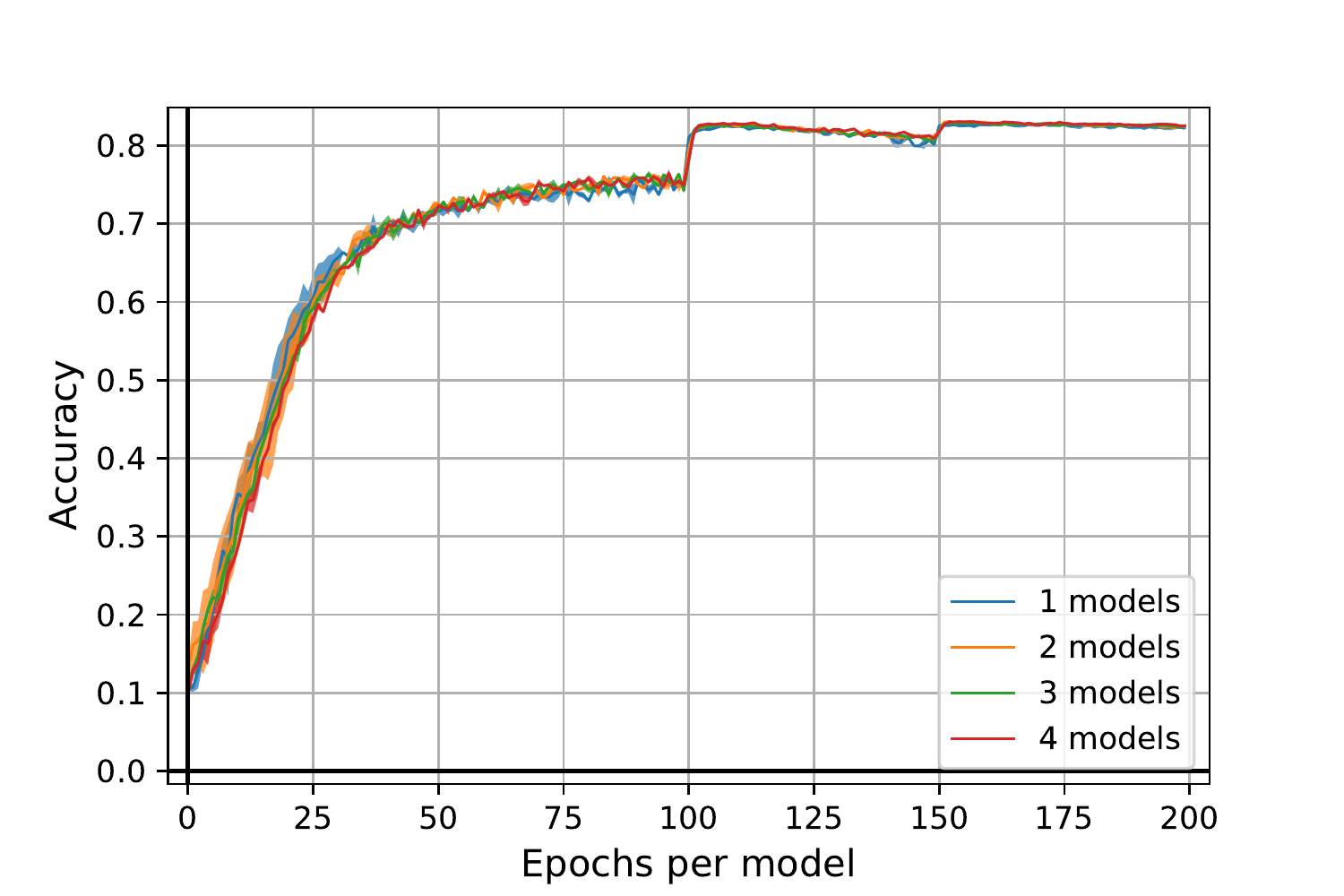}
    \includegraphics[width=0.24\textwidth]{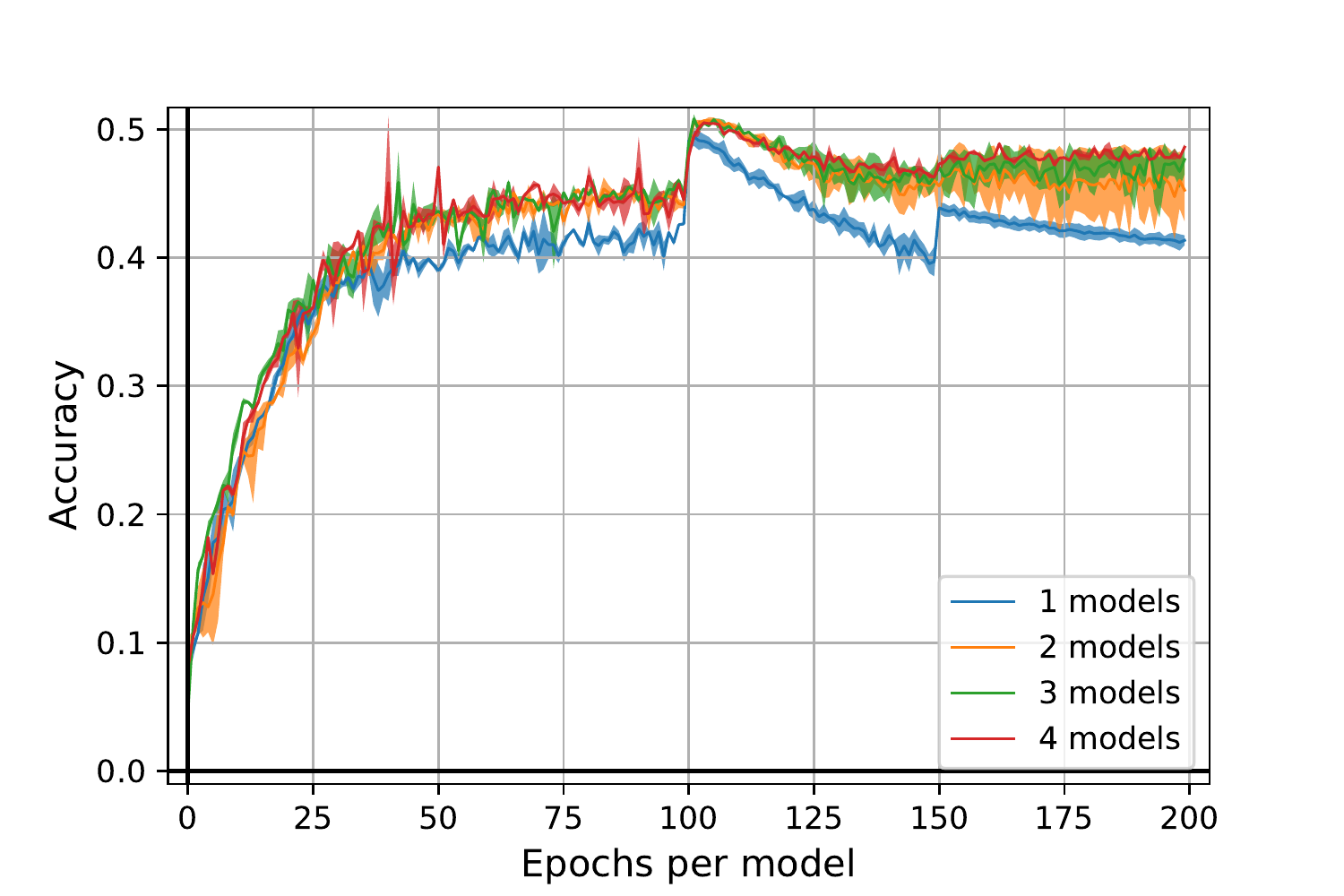}    \includegraphics[width=0.24\textwidth]{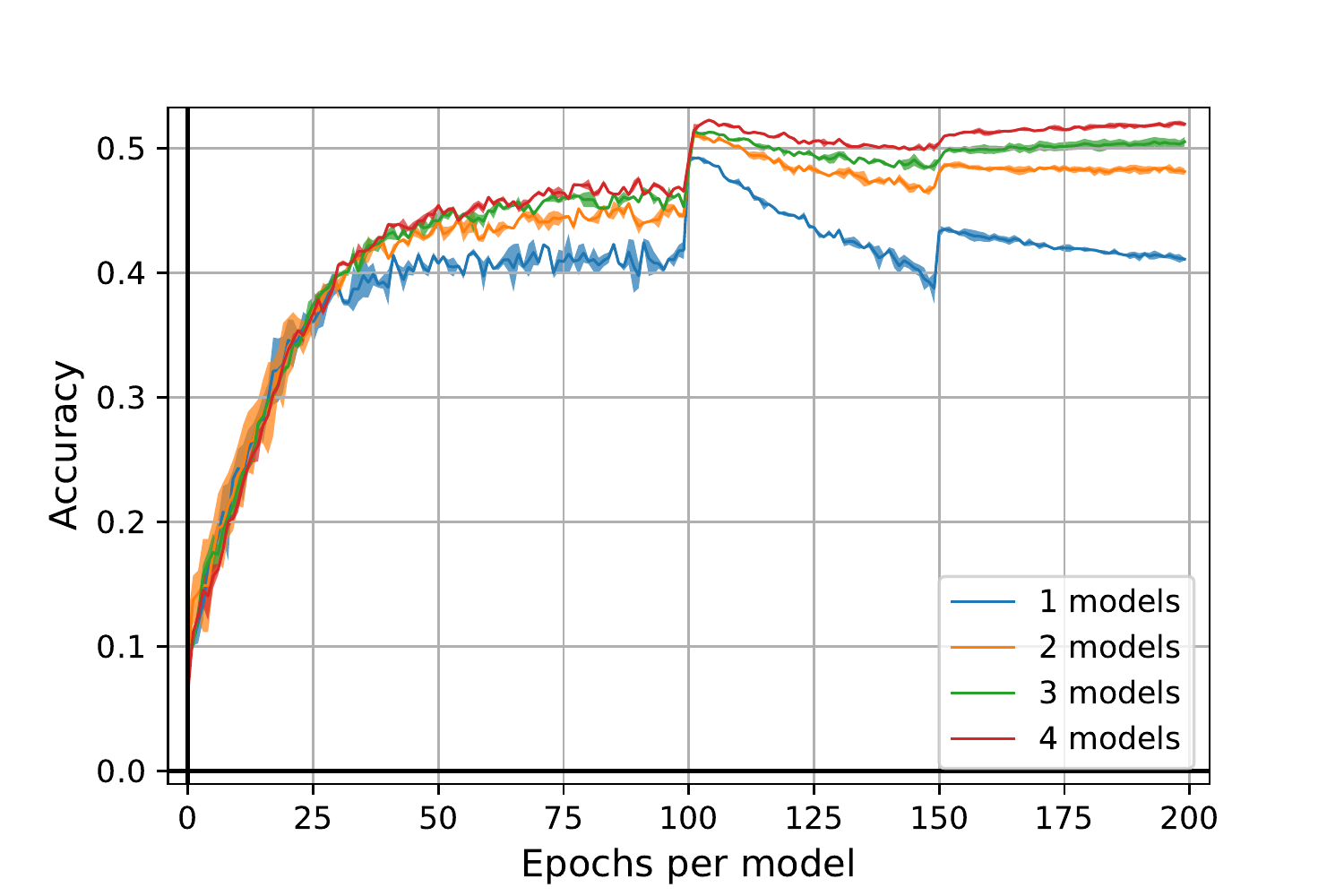}
    \caption{On left and middle-left: Standard accuracies over epochs with respectively no regularization and regularization set to $\alpha=0.001$. On middle right and right: Robust accuracies for the same parameters against PGD attack with $20$ iterations and $\varepsilon=0.03$.}
    \label{fig:xp-regularization}
\end{figure}

\subsection{Additional Experiments on WideResNet28x10}

We now evaluate our algorithm on WideResNet28x10~\cite{ZagoruykoK16} architecture. Due to computation costs, we limit ourselves to $1$ and $2$ models, with regularization parameter set to $0.001$ as in the paper experiments section. Results are reported in Figure~\ref{fig:xp-wideresnet}. We remark this architecture can lead to more robust models, corroborating the results from~\cite{gowal2020uncovering}.
\begin{figure*}[!ht]
\begin{center}

\vskip 0.15in
 \begin{minipage}[ht!]{0.39\textwidth}
 \begin{scriptsize}
\begin{tabular}{c|c|ccc} 
\textbf{ Models} & \textbf{Acc. }&\textbf{$\textrm{APGD}_\textrm{CE}$}& \textbf{$\textrm{APGD}_\textrm{DLR}$} & \textbf{Rob. Acc.} \\ \hline
 1 & $85.2\%$ &	$49.9\%$ & $50.2\%$ & $48.5\%$ \\ 
 2 & $\bm{86.0\%}$ & $\bm{51.5\%}$ & $\bm{52.1\%}$ & $\bm{49.6\%}$\\ 

\end{tabular}
\end{scriptsize}
  \end{minipage}\begin{minipage}[!ht]{0.61\textwidth}
\includegraphics[width=0.49\textwidth]{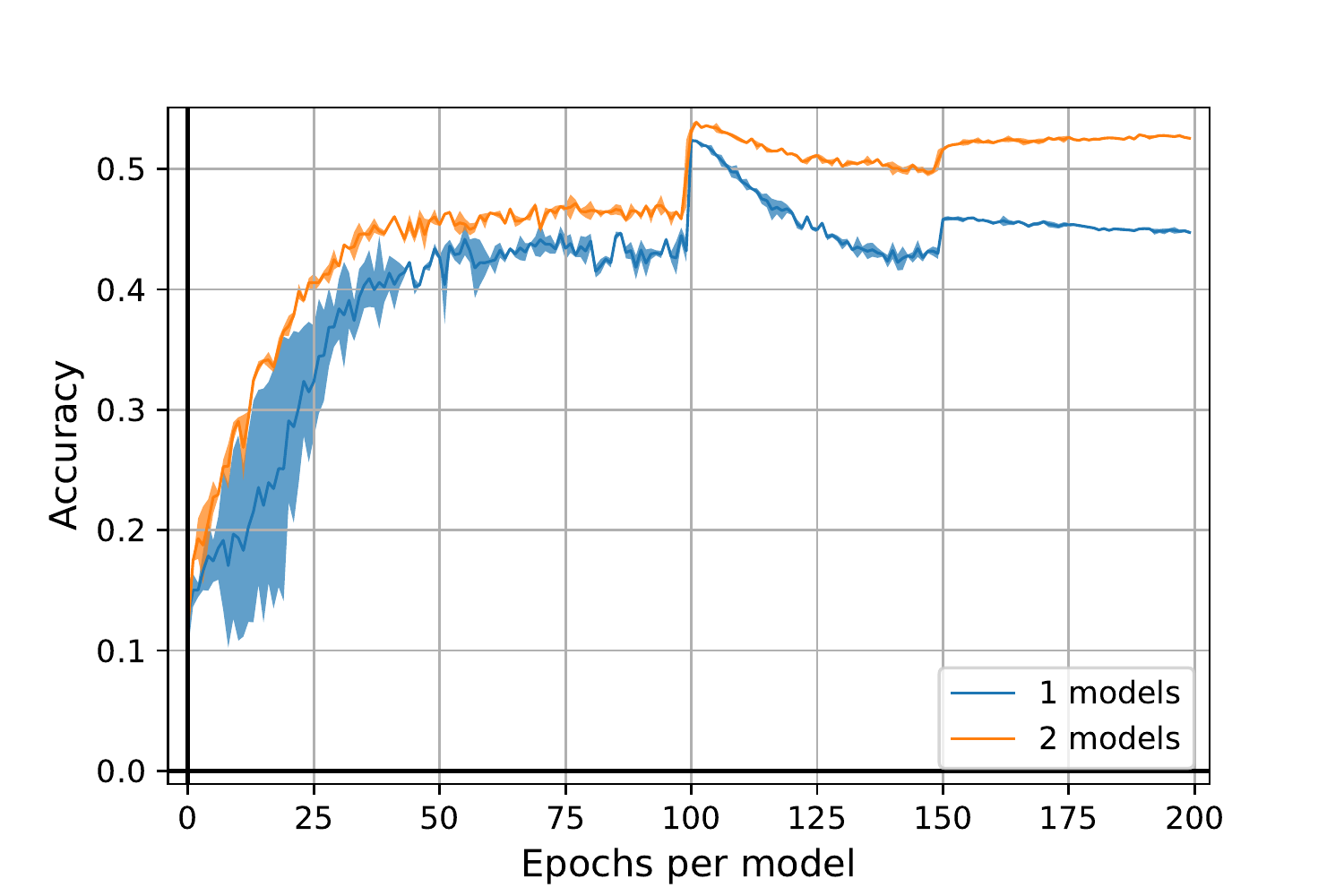}\includegraphics[width=0.49\textwidth]{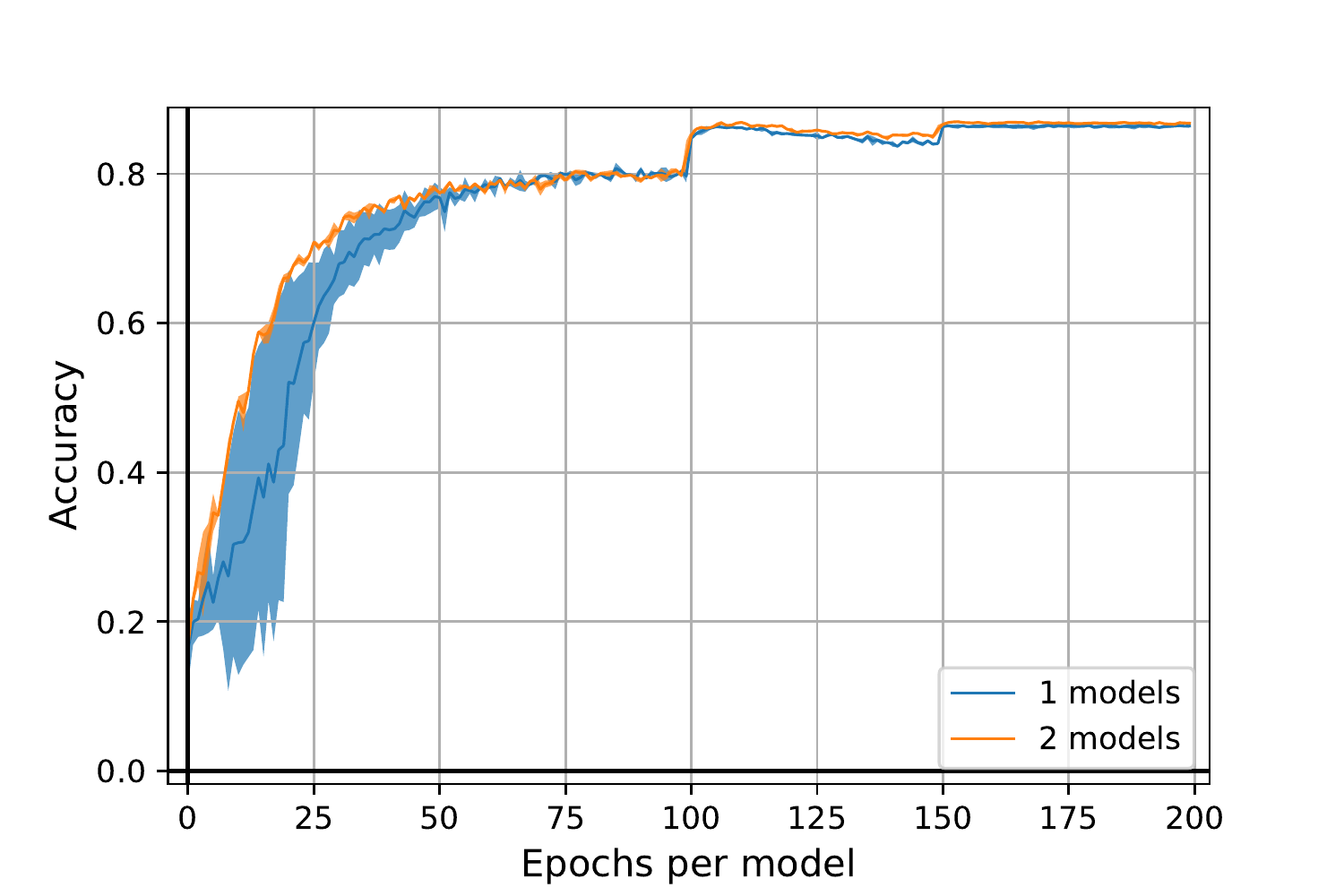} 
  \end{minipage}
  
\caption{On left: Comparison of our algorithm with a standard adversarial training (one model) on WideResNet28x10. We reported the results for the model with the best robust accuracy obtained over two independent runs because adversarial training might be unstable. Standard and Robust accuracy (respectively in the middle and on right) on CIFAR-10 test images in function of the number of epochs per classifier with $1$ and $2$ WideResNet28x10 models. The performed attack is PGD with $20$ iterations and $\varepsilon=8/255$.}
\label{fig:xp-wideresnet}
\end{center}
\vspace{-0.1cm}
\end{figure*}

\subsection{Overfitting in Adversarial Robustness}
We further investigate the overfitting of our heuristic algorithm. We plotted in Figure~\ref{fig:overfitting} the robust accuracy on ResNet18 with $1$ to $5$ models. The most robust mixture of $5$ models against PGD with $20$ iterations arrives at epoch $198$, \emph{i.e.} at the end of the training, contrary to $1$ to $4$ models, where the most robust mixture occurs around epoch $101$. However, the accuracy against AGPD with 100 iterations in lower than the one at epoch $101$ with global robust accuracy of $47.6\%$ at epoch $101$ and $45.3\%$ at epoch 198. This strange phenomenon would suggest that the more powerful the attacks are, the more the models are subject to overfitting. We leave this question to further works.

\begin{figure*}[!ht]
\begin{center}
\includegraphics[width=0.49\textwidth]{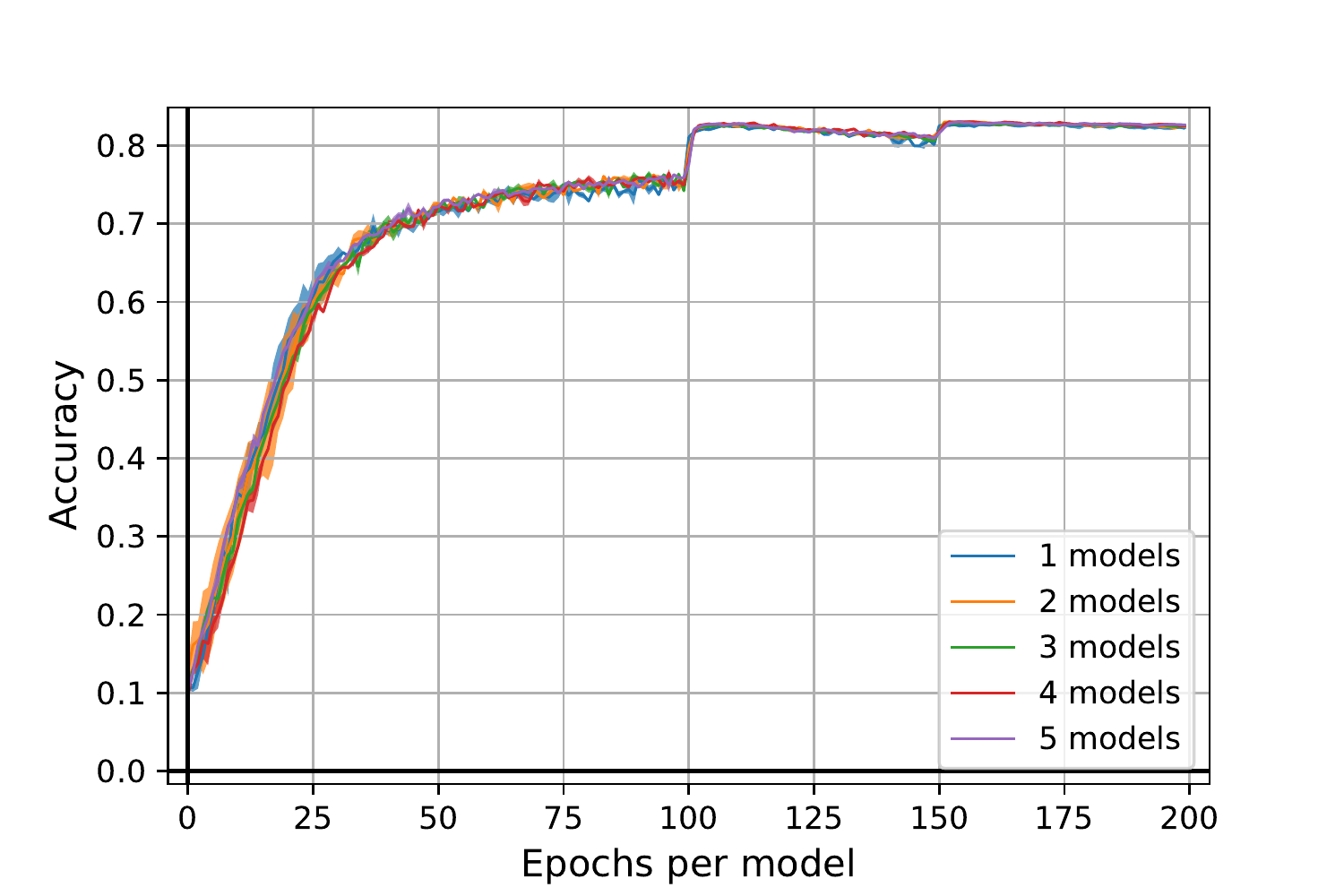}\includegraphics[width=0.49\textwidth]{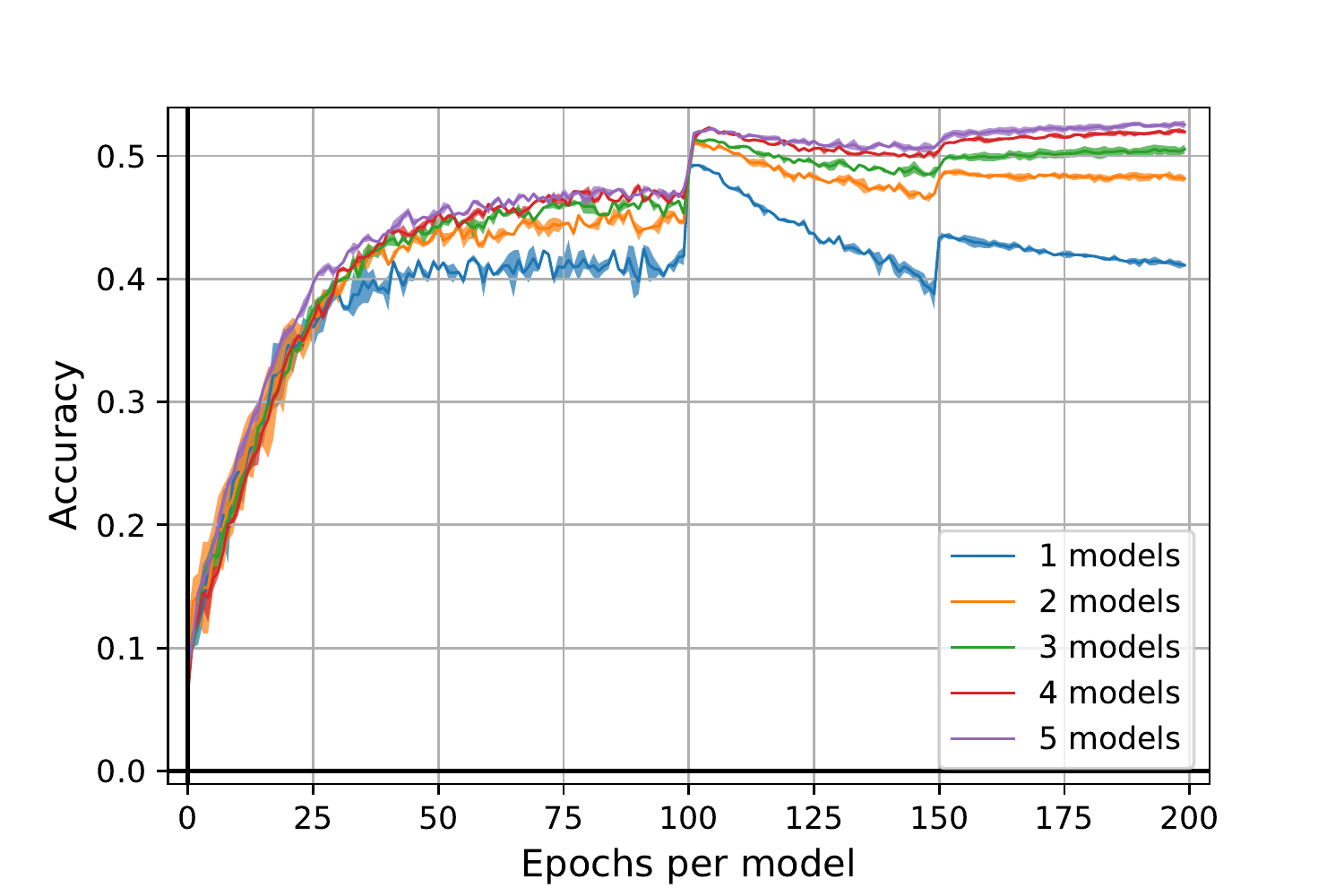} 

\caption{Standard and Robust accuracy (respectively on  left and on right) on CIFAR-10 test images in function of the number of epochs per classifier with $1$ to $5$ ResNet18 models. The performed attack is PGD with $20$ iterations and $\varepsilon=8/255$. The best mixture for $5$ models occurs at the end of training (epoch $198$).}
\label{fig:overfitting}
\end{center}
\vspace{-0.1cm}
\end{figure*}
\section{Additional Results}
\label{sec:complements}
\subsection{Equality of Standard Randomized and Deterministic Minimal Risks}
\begin{prop}
Let $\PP$ be a Borel probability distribution on $\mathcal{X}\times\mathcal{Y}$, and $l$ a loss satisfying Assumption~\ref{ass:loss}, then:
\begin{align*}
        \inf_{\mu\in\mathcal{M}^1_+(\Theta)} \risk(\mu) =\inf_{\theta\in\Theta} \risk(\theta)
\end{align*}
\end{prop}
\begin{proof}
It is clear that:         $\inf_{\mu\in\mathcal{M}^1_+(\Theta)} \risk(\mu) \leq \inf_{\theta\in\Theta} \risk(\theta)$. Now, let $\mu\in\mathcal{M}^1_+(\Theta)$, then:
\begin{align*}
    \risk(\mu)= \mathbb{E}_{\theta\sim\mu}(\risk(\theta))&\geq \essinf_\mu \mathbb{E}_{\theta\sim\mu} \left(\risk(\theta)\right)\\
    &\geq\inf_{\theta\in\Theta} \risk(\theta).
\end{align*}
where $\essinf$ denotes the essential infimum.
\end{proof}
We can deduce an immediate corollary. 
\begin{corollary}
Under Assumption~\ref{ass:loss}, the dual for randomized and deterministic classifiers are equal.
\end{corollary}

\subsection{Decomposition of the Empirical Risk for Entropic Regularization}

\begin{prop}
Let $\hat{\PP}:=\frac1N\sum_{i=1}^N \delta_{(x_i,y_i)}$. Let $l$ be a loss satisfying Assumption~\ref{ass:loss}. Then we have:
\begin{align*}
\frac{1}{N}\sum_{i=1}^N\sup_{x,~d(x,x_i)\leq\varepsilon}\mathbb{E}_{\theta \sim \mu}\left[l(\theta,(x,y))\right]=\sum_{i=1}^N\sup_{\QQ_i\in\Gamma_{i,\varepsilon}}\mathbb{E}_{(x,y)\sim \QQ_i,\theta \sim \mu}\left[l(\theta,(x,y))\right]
\end{align*}
where $\Gamma_{i,\varepsilon}$ is defined as : 
\begin{align*}
    \Gamma_{i,\varepsilon}:=\Big\{\QQ_i\mid~\int d\QQ_i=\frac{1}{N},~\int c_{\varepsilon}((x_i,y_i),\cdot) d\QQ_i=0\Big\}.
\end{align*}\end{prop}

\begin{proof}
This proposition is a direct application of Proposition~\ref{prop:dro_adv} for diracs $\delta_{(x_i,y_i)}$.
\end{proof}

\subsection{On the NP-Hardness of Attacking a Mixture of Classifiers}
In general, the problem of finding a best response to a mixture of classifiers is in general NP-hard. Let us justify it on a mixture of linear classifiers in binary classification: $f_{\theta_k}(x) = \langle \theta,x\rangle$ for $k\in [L]$ and $\bm{\lambda}=\mathbf{1}_L/L$. Let us consider the $\ell_2$ norm and $x=0$ and $y=1$. Then the problem of attacking $x$ is the following:
\begin{align*}
    \sup_{\tau,~\lVert \tau\rVert\leq\varepsilon} \frac{1}{L}\sum_{k=1}^L\mathbf{1}_{\langle \theta_k,\tau\rangle\leq0}
\end{align*}
This problem is equivalent to a linear binary classification problem on $\tau$, which is known to be NP-hard.
\subsection{Case of Separated Conditional Distribtions}
\begin{prop} Let $\mathcal{Y} = \{-1,+1\}$. Let $\PP\in\mathcal{M}^1_+(\mathcal{X}\times\mathcal{Y})$. Let $\varepsilon>0$. For $i\in\mathcal{Y}$, let us denote $\PP_i$ the distribution of $\PP$ conditionally to $y=i$. Let us assume that  $d_\mathcal{X}(\supp(\PP_{1+1}),\supp(\PP_{-1}))>2\varepsilon$. Let us consider the nearest neighbor deterministic classifier : $f(x) =  d(x,\supp(\PP_{+1}))-d(x,\supp(\PP_{-1}))$ and the $0/1$ loss $l(f,(x,y))=\mathbf{1}_{yf(x)\leq 0}$. Then $f$ satisfies both optimal standard and adversarial risks: $\risk(f)=0$ and $\riskadv^\varepsilon(f)=0$.
\end{prop}

\begin{proof} Let 
Let denote $p_i =\PP(y=i)$. Then we have
\begin{align*}
 \riskadv^\varepsilon(f)= p_{+1}\mathbb{E}_{\PP_{+1}}\left[\sup_{x',~d(x,x')\leq \varepsilon}\mathbf{1}_{ f(x')\leq 0}\right]+p_{-1}\mathbb{E}_{\PP_{-1}}\left[\sup_{x',~d(x,x')\leq \varepsilon}\mathbf{1}_{ f(x')\geq 0}\right]
\end{align*}
For $x\in \supp(\PP_{+1})$, we have, for all $x'$ such that $d(x,x')\neq 0$, $f(x')>0$, then: $\mathbb{E}_{\PP_{+1}}\left[\sup_{x',~d(x,x')\leq \varepsilon}\mathbf{1}_{ f(x')\leq 0}\right]=0$. Similarly, we have $\mathbb{E}_{\PP_{-1}}\left[\sup_{x',~d(x,x')\leq \varepsilon}\mathbf{1}_{ f(x')\geq 0}\right]=0$. We then deduce the result.
\end{proof}

\end{document}